\newif\ifarxiv

\arxivtrue

\documentclass[3p,sort]{elsarticle}
\usepackage[utf8]{inputenc}
\usepackage[T1]{fontenc}
\usepackage{amsmath,amssymb,amsthm}
\usepackage[defblank]{paralist}
\usepackage{booktabs}
\usepackage{multirow}
\usepackage{array}
\usepackage{tabularx}

\makeatletter
\def\NAT@spacechar{~}%
\makeatother
\usepackage{tikz}
\usetikzlibrary{decorations.pathmorphing}
\usetikzlibrary{decorations.pathreplacing}
\usetikzlibrary{shapes,automata}

\usepackage{color}
\usepackage{graphicx}
\definecolor{mygrey}{rgb}{0.9,0.9,0.9}
\definecolor{darkgreen}{RGB}{0,100,0}

\usepackage[pdfdisplaydoctitle,menucolor=orange!40!black,filecolor=magenta!40!black,urlcolor=blue!40!black,linkcolor=red!40!black,citecolor=green!40!black,colorlinks,pagebackref]{hyperref}

\usepackage{doi}
\usepackage[strings]{underscore}

\usepackage[sort&compress,capitalize]{cleveref}

\renewcommand*{\backref}[1]{}
\renewcommand*{\backrefalt}[4]{
  \ifcase #1 Not cited.%
  \or Cited on page~#2.%
  \else Cited on pages #2.%
  \fi%
}

\theoremstyle{definition}
\newtheorem{theorem}{Theorem}[section]
\newtheorem{corollary}[theorem]{Corollary}
\newtheorem{definition}[theorem]{Definition}
\newtheorem*{problem}{Problem}
\newtheorem{lemma}[theorem]{Lemma}
\newtheorem{proposition}[theorem]{Proposition}

\newtheorem{observation}[theorem]{Observation}
\newtheorem{rrule}[theorem]{Reduction Rule}
\newtheorem{constr}[theorem]{Construction}
\theoremstyle{remark}
\newtheorem*{remark}{Remark}

\crefname{constr}{Construction}{Constructions}
\crefname{step}{Step}{Steps}
\crefname{observation}{Observation}{Observations}
\crefname{proposition}{Proposition}{Propositions}
\crefname{remark}{Remark}{Remarks}
\crefname{rrule}{Reduction Rule}{Reduction Rules}
\crefname{prop}{Property}{Properties}
\crefname{figure}{Figure}{Figures}
\creflabelformat{prop}{(#2#1#3)}

\newcommand{\decprob}[3]{
  \begin{problem}[\textsc{#1}]\leavevmode
    \begin{compactdesc}
    \item[Input:] #2
    \item[Question:] #3
    \end{compactdesc}   
  \end{problem}
}

\newcommand{\hs}[1]{\ensuremath{#1}-\textsc{Hitting Set}}

\DeclareMathOperator{\petal}{petal}

\newcommand{\torso}{\ensuremath{\operatorname{torso}}}

\newcommand{\N}{\mathbb{N}}
\newcommand{\I}{\mathcal{I}}
\newcommand{\R}{\mathcal{R}}
\newcommand{\F}{\mathcal{F}}
\newcommand{\C}{\mathcal{C}}
\newcommand{\ffvd}{\textsc{\(\F\)-free Vertex Deletion}}
\newcommand{\sffvd}{\textsc{Secluded \(\F\)-free Vertex Deletion}}
\newcommand{\ssffvd}{\textsc{Small Secluded \(\F\)-free Vertex Deletion}}
\newcommand{\sssts}{\textsc{Small Secluded $s$\nobreakdash-$t$\nobreakdash-Separator}}

\newcommand{\ssp}{\textsc{Small Secluded~$\Pi$}}
\newcommand{\lsp}{\textsc{Large Secluded~$\Pi$}}
\newcommand{\fco}{\textsc{Fixed  Cardinality Optimization}}

\newcommand{\W}[1]{\ensuremath{\operatorname{W}[#1]}}

\newcommand{\poly}{\ensuremath{\operatorname{poly}}}

\newcommand{\NP}{\ensuremath{\text{NP}}}
\newcommand{\coNP}{\ensuremath{\text{coNP}}}
\newcommand{\nopk}{\ensuremath{\NP\subseteq \coNP/\text{poly}}}

\newcommand{\yes}{yes}

\newcommand{\thetitle}{The parameterized complexity of finding secluded solutions to some classical optimization problems on graphs}

\newcommand{\raproof}{($\Rightarrow$)}
\newcommand{\laproof}{($\Leftarrow$)}

\date{}

\begin{document}

\begin{frontmatter}

  \title{\thetitle{}\tnoteref{ipec}}
  \tnotetext[ipec]{A preliminary version of this work appeared in the proceedings of IPEC 2016 \citep{BFM+17}.  This version provides full proof details and additionally shows that \textsc{Small Secluded \(s\)-\(t\)-Separator} is fixed-parameter tractable parameterized by the combination of the solution size and the open neighborhood size (\cref{thm:fptsssts}).}

\author[1,2]{Ren\'{e}~van~Bevern\fnref{thx1}}
\ead{rvb@nsu.ru}
\fntext[thx1]{Results in \cref{sec:ffvd} were obtained under support of the Russian Science Foundation, grant~16-11-10041.}

\author[3]{Till~Fluschnik\corref{cor1}\fnref{thx2}}
\ead{till.fluschnik@tu-berlin.de}
\fntext[thx2]{Supported by the DFG, project DAMM (NI~369/13).}
\cortext[cor1]{Corresponding author.}

\author[4]{George~B.~Mertzios\fnref{thx2b}}
\ead{george.mertzios@durham.ac.uk}
\fntext[thx2b]{Supported by the EPSRC project EP/P020372/1.}

\author[3]{Hendrik~Molter\fnref{thx3}}
\ead{h.molter@tu-berlin.de}
\fntext[thx3]{Supported by the DFG, project DAPA (NI 369/12).}

\author[3,5]{{Manuel~Sorge}\fnref{thx4}}
\ead{sorge@post.bgu.ac.il}
\fntext[thx4]{Supported by the People Programme (Marie Curie Actions) of the European Union's Seventh Framework Programme (FP7/2007-2013) under REA grant agreement number 631163.11 and by the Israel Science Foundation (grant no. 551145/14).}

\author[6]{Ond\v rej~Such\'{y}\fnref{thx5}}
\ead{ondrej.suchy@fit.cvut.cz}
\fntext[thx5]{Supported by grant 17-20065S of the Czech Science Foundation.}

\address[1]{
\small{Department of Mechanics and Mathematics, Novosibirsk State University, Novosibirsk, Russian Federation,
}
}
\address[2]{
\small{Sobolev Institute of Mathematics of the Siberian Branch of the Russian Academy of Sciences, Novosibirsk, Russian Federation}
}
\address[3]{
\small{Institut f\"ur Softwaretechnik und Theoretische Informatik, TU~Berlin, Germany, 
}
}
\address[4]{
\small{School of Engineering and Computing Sciences, Durham University, Durham, UK,
}
}
\address[5]{
\small{Department of Industrial Engineering and Management, Ben-Gurion University of the Negev, Be'er Sheva, Israel
}
}
\address[6]{
\small{Faculty of Information Technology, Czech Technical University in Prague, Prague, Czech~Republic
}
}

\begin{abstract}
\looseness=-1
  This work studies the parameterized complexity of finding
  \emph{secluded} solutions to classical combinatorial optimization
  problems on graphs such as finding minimum \(s\)-\(t\) separators,
  feedback vertex sets, dominating sets, maximum independent sets, and vertex
  deletion problems for hereditary graph properties:
  Herein, one searches not only to minimize or maximize the size of
  the solution, but also to minimize the size of its
  neighborhood. This restriction has applications in secure routing
  and community detection. %
\end{abstract}


\end{frontmatter}
\ifarxiv{}\thispagestyle{empty}\fi{}

\section{Introduction}\label{sec:intro}
In many optimization problems on graphs, one searches for a minimum or
maximum cardinality subset of vertices and edges satisfying certain
properties, like a shortest $s$-$t$ path, a maximum independent set,
or a minimum dominating set.  \citet{Chechikjpp16} first studied the
problem of finding \emph{secluded} solutions, which additionally limit
the \emph{exposure} of the solution as measured by the size of the
neighborhood.  They motivate these problems by protecting sensitive
information that is sent through a network and potentially intercepted
by neighbors of its travel path.  However, given that there are
effective means of encrypting and signing sensitive information, the
following application seems more realistic: a convoy travelling from a
vertex~$s$ to a vertex~$t$ along an \(s\)-\(t\) path in a
transportation network can potentially be attacked from roads incident
to this path.  Thus, one arrives at the problem of finding an $s$-$t$
path with a small closed neighborhood.
Another motivation for limiting the exposure of solutions is the
search for segregated communities in social
networks~\citep{itoio05}. Here we search for dense subgraphs that
are exposed to few neighbors in the rest of the graph.
The constrained exposure models the concept of inter-cluster sparsity,
which states that communities have weak connections
to the rest of the network~\cite{Gae04}. In addition to
being a natural constraint in the above applications, restricting the
exposure of the solution may also yield more efficient
algorithms~\citep{hkmn09,huffnerks15,itoio05,khmn09}.

\citet{Chechikjpp16} and \citet{FominGKK16} previously studied
secluded paths and Steiner trees, respectively. Our aim in this paper
is to study the classical and parameterized complexity of secluded
variants of classical combinatorial optimization problems in graphs.

Following \citet{Chechikjpp16}, we measure the exposure of a
solution~$S$ by the size of the closed
neighborhood~$N_G[S] = S \cup \bigcup_{v \in S}N_G(v)$ of~$S$ in the
input graph~$G$. Given a predicate~$\Pi(G, S)$ that determines whether
$S$ is a solution for input graph~$G$, we study the following general
problem.
\decprob{Secluded $\Pi$}  {A graph~$G=(V,E)$ and an integer $k$.}  {Is there a
vertex subset $S\subseteq V$ such that $S$~satisfies~$\Pi(G, S)$ and $|N_G[S]|\leq k$?}

\noindent
In some cases, it may be necessary to control the size of the solution
and its neighborhood independently: For example, when routing a convoy
from~$s$ to~$t$ as above, we may simultaneously
aim to minimize its travel time, that is, the number of vertices on
its route, and to limit the exposure. Hence, another measure
for the exposure of the solution is the size of the open neighborhood
$N_{G}(S) = N_{G}[S] \setminus S$.  We thus introduce and study the complexity
of the following problem.
\decprob{Small (Large) Secluded $\Pi$}
{A graph~$G=(V,E)$ and two integers $k,\ell $.} {Is there a vertex subset
  $S\subseteq V$ such that $S$
  satisfies~$\Pi(G, S)$, $|S|\leq k$, and $|N_{G}(S)|\leq \ell$
  (or $|S|\geq k$, and $|N_{G}(S)|\leq \ell$, respectively)?}

\paragraph{Our contributions} We study \textsc{Secluded \(\Pi\)} and
\textsc{Small Secluded \(\Pi\)} in the framework of parameterized
complexity, a framework allowing for a fine-grained complexity
analysis and for proving the effectivity of polynomial-time data
reduction (we give formal definitions in \cref{sec:prelim}): a problem
is \emph{fixed-parameter tractable} with respect to some
parameter~\(k\) if it can be solved in \(f(k)\cdot n^{O(1)}\)~time,
where $n$ is the input size.
Thus, for small parameters~\(k\), fixed-parameter algorithms can
potentially lead to efficient algorithms for NP-hard problems.  In
contrast, if a problem is W[1]-hard with respect to a parameter~\(k\),
then it is presumably not fixed-parameter tractable with respect to
that parameter.  Our results are summarized in~\Cref{results-table}.
\renewcommand{\arraystretch}{1.1}
\newcommand{\tabref}[1]{\footnotesize{(Thm.~\ref{#1})}}
\newcommand{\tabrefx}[2]{\footnotesize{(#1~\ref{#2})}}
\newcommand{\tabrefxx}[2]{\footnotesize{(Cor.~\ref{#1}/Thm.~\ref{#2})}}
\newcommand{\tabrefxxx}[2]{\footnotesize{(Thm.~\ref{#1}/Thm.~\ref{#2})}}
\begin{table}[t]
  \centering
  \begin{tabular}{@{}lllllll@{}}
    \toprule
                          & \multicolumn{2}{l}{Complexity}                & \multicolumn{4}{l}{Parameterized Complexity / Kernelization}                           \\ 
                          & \emph{Secluded}                                       & \emph{Small Secl.}               & \emph{Secluded}      & \multicolumn{3}{l}{\emph{Small Secl.}}         \\
    Problem               &                                                &                           & $k$           & $k$          & $\ell$       & $k+\ell$ \\
    \midrule%
    $s$-$t$ Separator     & P                                   & \NP-c.                    	& P             & W[1]-h.       			& W[1]-h.       			& FPT/noPK     \\
			  & \tabref{thm:sstspolytime} 		& \tabref{thm:ssstshard}	& 		& \tabref{thm:ssstshard}		& \tabref{thm:ssstshard}		& \tabrefxxx{thm:fptsssts}{thm:nopksssts}	\\ \\
    $q$-Dom.\ Set,     & \NP-c.                  	& \NP-c.\textsuperscript{*} 		& W[2]-h.       & $\rightarrow$ 			& $\rightarrow$ 			& W[2]-h.    \\
    $2p \le q$         & \tabref{thm:q-dom-NPh} 		& 				& \tabref{thm:q-domW2}			& 		& 						& \tabrefx{Cor.}{cor:q-dom-hard}	\\ \\
    $q$-Dom.\ Set,       & \NP-c.                 	& \NP-c.\textsuperscript{*} 	& FPT/noPK     			 	& W[2]-h.\textsuperscript{*}             & ?             & FPT/noPK   \\
    $2p > q$		  & \tabref{thm:q-dom-NPh} 		& 				& \tabrefxx{cor:q-dom-FPT}{thm:q-dom-NPh}	& & & \tabrefxx{cor:q-dom-hard}{thm:q-dom-FPT}	\\ \\
    $\mathcal{F}$-free  & \NP-c.                 	& \NP-c.\textsuperscript{*} 	& FPT/PK        & ?            				& ?            				& FPT/?      \\
    Vertex Deletion		  & \tabref{thm:sffcdnphard} 		& 				& \tabref{thm:sffvdpolyker}		& 		& 					& \tabref{thm:ssffvd}	\\ \\
    Feedback   & \NP-c.                              & \NP-c.\textsuperscript{*} 	& FPT/PK        		& ?             	& W[1]-h.       			& ? / ?      \\
    Vertex Set & \tabref{thm:sfvsisnphard} 		& 				& \tabref{thm:sfvspolyker}		& 		& \tabref{thm:ssfvsW1hard}		& 	\\
    
    \midrule
			  &		& \emph{Large Secl.} & 				& \multicolumn{3}{l}{\emph{Large Secl.}} \\
    Independent Set     &  					& \NP-c.\textsuperscript{*} &                           & $\rightarrow$ 		& $\rightarrow$ 			& W[1]-h.                    \\    
			&  		& 				& 						& 	& 	& \tabref{thm:lsishard}	\\
    \bottomrule
  \end{tabular}
  \caption{Overview of our results. PK stands for polynomial kernel. The results marked by an asterisk follow by a straightforward reduction from the non-secluded variant. The complexity for combinations marked by an arrow are resolved by a stronger result in a farther column.
  }
  \label{results-table}
\end{table}

We analyze the impact of the parameter~\(k\) on the complexity of
\textsc{Secluded $\Pi$} and the impact of the parameters~\(k\)
and~\(\ell\) on the complexity of \textsc{Small Secluded $\Pi$}.  The
predicates $\Pi(G, S)$ that we study are
\begin{itemize}
\item $s$-$t$~\textsc{Separator}, %
\item \textsc{Feedback Vertex Set} (\textsc{FVS}), %
\item $\mathcal{F}$\textsc{-free Vertex Deletion} ($\mathcal{F}$-\textsc{FVD}) (for an arbitrary finite family
$\mathcal{F}$ of graphs), encompassing \textsc{Cluster Vertex Deletion}, for example, %
and
\item \textsc{Independent Set} (\textsc{IS})%
.
\end{itemize}
Perhaps surprisingly, we find that \textsc{Secluded $s$-$t$-Separator}
is polynomial-time solvable, whereas \textsc{Small Secluded $s$-$t$
  Separator} is NP-complete. The remaining secluded problem variants are
NP-complete. For them, roughly speaking, we prove that fixed-parameter
tractability results for $\Pi$ parameterized by the solution size
carry over to \textsc{Secluded $\Pi$} parameterized by the closed neighborhood size~$k$.  For
\textsc{Small Secluded $\Pi$} parameterized by the open neighborhood size~$\ell$, however, we
mostly obtain W[1]-hardness.  On the positive side, for \textsc{Small
  Secluded} $\mathcal{F}$-\textsc{FVD} and \textsc{Small Secluded \(s\)-\(t\)-Separator}, we prove fixed-parameter
tractability when parameterized by~$k + \ell$.

 We also study, for two integers $p < q$, the $p$-secluded version of
\textsc{$q$-Dominating Set} (\textsc{$q$-DS}): a vertex set~$S$ is a
\emph{$q$-dominating set} if every vertex of~$V\setminus S$ has
distance at most $q$ to some vertex in~$S$. By $p$\nobreakdash-secluded we
mean that we upper bound the size of the distance-$p$-neighborhood of
the solution~$S$. This problem admits a complexity
dichotomy: Whenever $2p > q$, \textsc{(Small) $p$-Secluded
  $q$-Dominating Set} is fixed-parameter tractable with respect to $k$
(with respect to $k + \ell$), but it is W[2]-hard otherwise.

Finally, we also study the possibility for effective polynomial-time
data reduction.  We observe that the polynomial-size problem kernels
for \textsc{Feedback Vertex Set} and \textsc{$\mathcal{F}$-free Vertex
  Deletion} carry over to their \textsc{Secluded} variants, but
otherwise we obtain mostly absence of polynomial-size problem kernels
unless the polynomial hierarchy collapses.

\newcommand{\pSP}{\textsc{Secluded Path}}
\newcommand{\pST}{\textsc{Secluded Steiner Tree}}
\paragraph{Related work} %
\pSP\ and \pST\ were introduced and
proved NP-complete by~\citet{Chechikjpp16}. They
obtained approximation algorithms for both problems with approximation
factors related to the maximum degree.
They also showed that \textsc{Secluded Path} is fixed-parameter
tractable with respect to the maximum vertex degree of the input
graph, whereas vertex weights lead to NP-hardness for maximum degree
four.

\citet{FominGKK16} studied the parameterized complexity of
\pSP\ and \pST, showing that both are fixed-parameter tractable even
in the vertex-weighted setting. Furthermore, they showed that
\textsc{Secluded Steiner Tree} is fixed-parameter tractable with
respect to $k - s + p$, where $p$ is the number of terminals, %
$k$ is the desired size of the
closed neighborhood of the solution, and $s$ is the size of an optimum
Steiner tree. On the other hand
this problem is co-W[1]-hard when parameterized by $k - s$ only~\cite{FominGKK16}.

The small secluded concept can be found in the context of separator %
problems in graphs~\citep{Marx06,FominGK13}.  
\citet{Marx06} introduced the \textsc{Cutting $k$ Vertices} problem, which asks, given a graph~$G=(V,E)$ and two integers~$k\geq1$ and $\ell\geq0$, whether there is a non-empty set~$S\subseteq V$ such that~$|S|\leq k$ and $|N_G(S)|= \ell$.
It follows from the work of \citet{BuiJ92}  that the problem is NP-hard and \citet{Marx06} proved that the problem is W[1]-hard with respect to~$k+\ell$.
Moreover, for the problem variant where the set $S$ is required to induce a connected subgraph in~$G$, he proved that it becomes fixed-parameter tractable when parameterized by~$k+\ell$, while staying W[1]-hard with respect to~$k$ or~$\ell$.
\citet{FominGK13} studied the variant of \textsc{Cutting $k$ Vertices} where~$|N_G(S)|\leq \ell$, thus resembling our small secluded concept.
They proved that this variant is W[1]-hard with respect to~$k$, but becomes fixed-parameter tractable when parameterized by~$\ell$. 
{As to the latter, we remark that for none of our studied small secluded problems we observed fixed-parameter tractability with respect to~$\ell$.}
Somewhat surprisingly, if it is additionally required that~$S$ has to contain a predefined vertex~$s\in V$, then this variant becomes W[1]-hard with respect to~$\ell$ while staying fixed-parameter tractable when parameterized by~$k+\ell$~\cite{FominGK13}.

The concept of \emph{isolation} states that the solution should have
few edges to the rest of the graph and was originally introduced for
finding cliques~\cite{itoio05}. Isolation was subsequently explored
also for more general definitions of dense
subgraphs~\citep{hkmn09,huffnerks15,itoio05,khmn09}. Chiefly
the constraint that the vertices in the solution shall have
maximum/minimum/average outdegree bounded by a parameter was
considered~\citep{hkmn09,itoio05,khmn09}, leading to various
parameterized tractability and hardness results. Also the overall
number of edges outgoing the solution has been studied
recently~\citep{huffnerks15}. Finding isolated vertices without constraint on their topology was already studied by \citet{DowneyEFPR03}. 

\ssp\ and \lsp\ can be seen as special cases of \fco~\cite{Bru+06,CCC06,Cai08,KS15}. Hence, we can derive some corollaries for secluded problems from results from the literature on \fco, see below.

\paragraph{Preliminary observations} Concerning the classical computational complexity, the \textsc{Small (Large) Secluded} variant of a problem is at least as hard as the nonsecluded problem, by a simple reduction in which we set $\ell =n$, where $n$ denotes the number of vertices in the graph.
Since this reduction is a parameterized reduction with respect to~$k$,
parameterized hardness results for this parameter transfer,
too. Furthermore, observe that hardness also transfers from
\textsc{Secluded $\Pi$} to \textsc{Small Secluded $\Pi$} for all
problems $\Pi$, since \textsc{Secluded $\Pi$} allows for a
parameterized Turing 
reduction to \textsc{Small Secluded $\Pi$}: try out
all $k'$ and $\ell'$ with $k=k'+\ell'$.  
\begin{observation}
\label{obs:reducibility}
\textsc{Secluded $\Pi$} parameterized by $k$ is parameterized Turing reducible to \textsc{Small Secluded $\Pi$} parameterized by $(k+\ell)$ for all predicates $\Pi$.
\end{observation}
\noindent
Additionally, many
tractability results (in particular polynomial time solvability and
fixed-parameter tractability) transfer from \textsc{Small
  Secluded~$\Pi$} parameterized by $(k+\ell)$ to \textsc{Secluded
  $\Pi$} parameterized by $k$.
Therefore, for the \textsc{Small (Large) Secluded} variants of the problems the interesting cases are those where the base problem is tractable (deciding whether input graph~$G$ contains a vertex set~$S$ of size~$k$ that satisfies $\Pi(G, S)$) or where the size~$\ell$ of the open neighborhood is a parameter.

\ssp\ can be solved using techniques for \fco~\cite{Bru+06,Cai08},
where we are given an objective function~$\phi$ over subsets of some
universe (usually via an oracle or efficient algorithm) and we seek to
optimize it over all $k$-element subsets: Define the universe as the
vertex set in the input graph, and set $\phi(S) = \infty$ whenever
$\Pi(G, S)$ is not satisfied. Otherwise $\phi(S) = |N(S)|$. Iterating
over all $i = 1, \ldots, k$, minimizing~$\phi$ for $i$-vertex
subgraphs, and checking whether a solution $S$ with
$\phi(S) \leq \ell$ exists thus solves \ssp.

In \ssp\ no vertex of degree greater than $k + \ell$ can be part of the
solution. Hence, the above formulation makes \ssp\ amenable to the
random separation framework for \fco\ in graphs with bounded
degree~\citep{CCC06}: Roughly, we can restrict the domain of~$\phi$ to
the set of vertices of degree smaller than~$k + \ell$, effectively
removing large-degree vertices from the graph, while their
contribution to satisfying~$\Pi$ and to the neighborhoods of other
vertices is still present in~$\phi$. Applying then a fixed-parameter
algorithm with respect to the maximum degree for \fco\ given by \citet[Theorem 4]{CCC06}, we can thus derive fixed-parameter
algorithms with respect to $k + \ell$ for some \ssp\ problems. Herein,
$\Pi(G, S)$ must be computable in fixed-parameter time and fulfill
certain conditions which roughly state that subgraphs of~$G$ that
fulfill $\Pi$ and have a certain distance from each other in~$G$ can
be combined via the disjoint union into one subgraph
fulfilling~$\Pi$. For example, finding secluded subgraphs of fixed
minimum degree, subgraphs of even degree, subgraphs that induce
matchings, and subgraphs of fixed diameter is fixed-parameter tractable
with respect to $k + \ell$ through that approach. Furthermore, as a
special case we obtain that, if $\Pi(G, S)$ is fixed-parameter
tractable with respect to $k + \ell$ and each vertex subset~$S$
satisfying~$\Pi(G, S)$ is connected, then \ssp\ is fixed-parameter
tractable with respect to~$k + \ell$.

Note that a similar strategy as above does not work for \lsp, because
we cannot bound the maximum degree in the same fashion: Neighbors of a
solution vertex may be taken into the (arbitrarily large) solution to
keep the neighborhood of the solution small.

\paragraph{Organization}
We give basic definitions from graph theory and parameterized
complexity theory in \cref{sec:prelim}.  Each of the subsequent
sections is dedicated to one of the studied problems.  We study
\textsc{\(s\)-\(t\)-Separator} in \cref{sec:sts},
\textsc{$q$-Dominating Set} in \cref{sec:dom}, \textsc{\(\F\)-free
  Vertex Deletion} in \cref{sec:ffvd}, \textsc{Feedback Vertex Set} in
\cref{sec:fvs}, and \textsc{Independent Set} in \cref{sec:is}.
\cref{sec:summary} summarizes results and gives directions for future
research.

\section{Preliminaries}
\label{sec:prelim}

We use standard notions from parameterized complexity~\cite{DowneyF13,FlumG06,Nie06,CyganFKLMPPS15} 
and graph theory~\cite{Diestel10,West00}. We use $[p]$ to denote the set $\{1,\ldots,p\}$.
	
\paragraph*{Graph theory}

Let~$G = (V, E)$ be an undirected graph (all graphs in this paper are undirected). 
We denote by $V(G)$ the vertex set of~$G$ and by~$E(G)$ the edge set of~$G$. 
For a vertex set $W \subseteq V (G)$ (edge set $F \subseteq E(G)$), we denote by~$G[W]$ ($G[F]$) the subgraph of~$G$ \emph{induced} by the vertex set~$W$ (edge set $F$), respectively. We also denote by $G-W$ the graph $G[V(G)\setminus W]$.

A graph with vertex set~$\{v_0,\ldots,v_x\}$ and edge set~$\{\{v_{i-1},v_i\}\mid i\in[x]\}$ is called a path (with endpoints $v_0$ and $v_x$, also referred to as $v_0$-$v_x$~path).
The length of a path is the number of edges.
A graph with vertex set~$\{v_0,\ldots,v_x\}$ and edge set~$\{\{v_{i-1},v_i\}\mid i\in[x]\}\cup \{v_x,v_0\}$ is called a cycle (of length $x+1$).

We denote $d_G(u,v)$ the \emph{distance} between vertices~$u$ and~$v$ in~$G$, that is, the number of edges of a shortest $u$-$v$ path in $G$. For a set $V'$ of vertices and a vertex~$v \in V$ we let the distance of $v$ from $V'$ be $d_G(v,V'):=\min\{d_G(u,v) \mid u \in V'\}$. We use $N^d_G[V']=\{v \mid d_G(v,V') \le d\}$ and $N^d_G(V') = N^d_G[V']\setminus V'$ for any $d \ge 0$ (hence $N^0_G(V') = \emptyset$). We omit the index if the graph is clear from context and also use $N[V']$ for $N^1[V']$ and $N(V')$~for $N^1(V')$. If $V'=\{v\}$, then we write $N^d[v]$ in place of $N^d[\{v\}]$.
The \emph{diameter} of a graph~$G$ is the maximum distance between $v$ and $w$ over all $v, w\in V(G)$. 

A subset $V'\subseteq V(G)$ is called an $s$-$t$ separator in~$G$ for two distinct vertices~$s$ and~$t$ if there is no $s$-$t$~path in~$G-V'$.
An $s$-$t$ separator $V'$ is called minimal if for all $V''\subsetneq V'$ holds that there is an $s$-$t$~path in $G-V''$.

\paragraph*{Parameterized complexity}

Parameterized complexity has been introduced to more effectively but
optimally solve NP-hard problems: one accepts the apparently
inevitable combinatorial explosion in algorithms for NP-hard problems,
yet decouples it from the input size and limits it to one aspect of
the problem---some small \emph{parameter}.  The instances~\((x,k)\) of
a \emph{parameterized problem}~$P\subseteq \Sigma^*\times\mathbb N$ consist
of an input~\(x\) and a parameter~\(k\).  A parameterized problem~$P$
is \emph{fixed-parameter tractable (FPT)} if, for every
$(x,k)\in \Sigma^*\times \N$, it can be decided
in~$f(k)\cdot |x|^{O(1)}$ time whether~$(x,k)\in P$, where~$f$ is an
arbitrary computable function only depending on~$k$.  The
fixed-parameter tractable parameterized problems form the
parameterized complexity class FPT.

Parameterized complexity provides means to show intractability: There
is a hierarchy of parameterized complexity classes
FPT\({}\subseteq{}\)W[1]\({}\subseteq{}\)W[2]\({}\subseteq\dots\subseteq{}\)W[P],
where all inclusions are conjectured to be strict.  
A \emph{parameterized reduction} from parameterized problem $P_1$ to parameterized problem $P_2$ is 
an algorithm that
maps an instance~\((x,k)\) of~\(P_1\) to an instance~$(x',k')$
of~\(P_2\) in time~$f(k)\cdot\poly(|x|)$ such that $(x,k)$ and $(x',k')$ 
are equivalent and $k'\leq g(k)$, where
\(f\)~and~\(g\) are arbitrary computable functions only depending
on~$k$.
We say that an instance~$(x,k)$ of parameterized problem~$P_1$ 
is \emph{equivalent} with an instance~$(x',k')$ of parameterized 
problem~$P_2$ if $(x,k)\in P_1\iff(x',k')\in P_2$.
A \emph{polynomial-parameter transformation} is a parameterized reduction 
for which both $f$ and $g$ are polynomial. 
Note that such a reduction is also a polynomial-time many-one reduction 
from~$P_1$ to~$P_2$ considered in classical complexity theory (also called Karp reduction).
A parameterized problem~$P_2$ is W[\(t\)]\emph{-hard} if there
is a parameterized reduction from each problem~\(P_1\in{}\)W[\(t\)] to~\(P_2\).
No W[\(t\)]-hard problem is fixed-parameter tractable unless FPT\({}={}\)W[\(t\)].

A \emph{parameterized Turing reduction}
from parameterized problem $P_1$ to parameterized problem $P_2$ is 
an algorithm that decides whether~$(x,k)\in P_1$ in~$f(k)\cdot |x|^{O(1)}$ time if it is  
provided with access to an oracle that can decide instances~$(x',k')$
of~\(P_2\) with $k'\leq g(k)$ in constant time.
Here \(f\)~and~\(g\) are arbitrary computable functions only depending
on~$k$.

\paragraph*{Problem kernelization}
Parameterized complexity introduced the concept of
kernelization to measure the effect of polynomial-time data
reduction.  
A \emph{kernelization} is a polynomial-parameter transformation of 
a parameterized problem~$P$ to itself, transforming an instance $(x,k)$ of~$P$
into an instance~$(x',k')$ (the \emph{kernel}), such that 
$|x'|+k'\leq f(k)$ for some computable function~$f$ only depending
on~$k$.  We call $f$ the \emph{size} of the kernel.  If
$f\in k^{O(1)}$, then we say that $P$ admits a polynomial-size kernel.
A decidable parameterized problem is fixed-parameter
tractable if and only if it admits a kernel~\cite{CyganFKLMPPS15}.
 
Using the following cross-composition method, one can show that a
problem does not have polynomial-size kernels unless the
polynomial-time hierarchy collapses to the third level.  

An equivalence
relation~$\R$ on the instances of some problem~$L$ is a
\emph{polynomial equivalence relation} if
\begin{itemize}[(i)]
 \item one can decide for any two instances in time polynomial in their sizes whether they belong to the same equivalence class, and
 \item for any finite set~$S$ of instances, $\R$ partitions the set into at most~$(\max_{x \in S} |x|)^{O(1)}$ equivalence classes. %
\end{itemize}
An \emph{OR-cross-composition} of an \NP-hard problem~$L$ into a
parameterized problem~$P$ (with respect to a polynomial equivalence
relation~$\R$ on the instances of~\(L\)) is an algorithm that takes
$\ell$ $\R$-equivalent instances~$x_1,\ldots,x_\ell$ of~$L$ and
constructs in time polynomial in $\sum_{i=1}^\ell |x_i|$ an instance
$(x,k)$ of~\(P\) such that
\begin{itemize}
\item $k$ is polynomially upper-bounded in $\max_{1\leq i\leq \ell}|x_i|+\log(\ell)$ and 
\item $(x,k)\in P$ if and only if there is at least one~$\ell'\in[\ell]$ such that $x_{\ell'}\in L$. 
\end{itemize}
If an \NP-hard problem~\(L\) OR-cross-composes into a parameterized
problem~$P$, then~$P$ does not admit a polynomial-size kernel, unless $\NP\subseteq \coNP/\poly$~\cite{CyganFKLMPPS15},
which would cause a collapse of the polynomial-time hierarchy to the third
level.

\section{\boldmath\(s\)-\(t\)-Separator}\label{sec:sts}
In this section, %
we show that \textsc{Secluded $s$-$t$-Separator} is in P, while \textsc{Small Secluded $s$-$t$-Separator} is \NP-hard and $W[1]$-hard with respect to the size of the open neighborhood and with respect to the size of the solution. 
Moreover, we show that, parameterized by the sum of the sizes of the
open neighborhood and the solution, the problem is fixed-parameter
tractable yet does not allow for polynomial-size kernels.

\subsection{Secluded \(s\)-\(t\)-Separator}\label{ssec:ssts}

In this subsection we show that the following problem can be solved in polynomial time.
\decprob{Secluded $s$-$t$-Separator}
	{A graph~$G=(V,E)$, two distinct vertices $s,t\in V$, and an integer $k$.}
	{Is there an $s$-$t$~separator $S\subseteq V\setminus \{s,t\}$ such that $|N_G[S]|\leq k$?}

\begin{theorem}%
  \label{thm:sstspolytime}
  \textsc{Secluded $s$-$t$-Separator} can be solved in polynomial time.
\end{theorem}

In order to prove~\cref{thm:sstspolytime}, we show that it is enough
to compute a separator of size~\(k\) in the third power of a
graph that is obtained from~\(G\) by adding two new terminals
to~\(s\) and~\(t\):

\begin{definition}
 \label{def:xthpower}
For $x \in \N$ the $x$-th power of a graph $G=(V,E)$ is a  
 graph \(G'=(V,E')\) where for each pair of distinct vertices $u, v \in V$ we have $\{u,v\} \in E'$ if and only if $d_{G}(u,v) \le x$.
\end{definition}

That is, the $x$-th power of a graph~$G$ is obtained by adding edges
between vertices that are at distance at most~$x$ in~$G$.  
This can be done in polynomial time.  

\begin{lemma}
 \label{lem:sepsinthirddpower}
 Let $G=(V,E)$ be an undirected graph with two distinct
 vertices~$s,t\in V$.  Let $G''$ be the graph obtained from $G$ by
 adding two vertices $s',t'$ and two edges $\{s',s\},\{t,t'\}$.  Then
 there is an $s$-$t$-separator $S$ in $G$ with $|N[S]| \le k$ if and
 only if there is an $s'$-$t'$-separator~$S'$ with~$|S'| \le k$ in the
 third power~$G'$ of~$G''$.
\end{lemma}

\begin{proof}

    \raproof{} Let $S$ be an $s$-$t$-separator in $G$ with $|N[S]| \le k$. Observe that $S$ then also constitutes an $s'$-$t'$-separator in $G''$ as every path in $G''$ from $s'$ must go through~$s$ and every path to $t'$ must go through~$t$. We claim that $S'= N[S]$ is an $s'$-$t'$-separator in $G'$. Suppose for contradiction that there is an $s'$-$t'$ path $P= p_0,p_1, \ldots, p_q$ in $G' - S'$. Let $A'$ be the set of vertices of the connected component of $G''-S$ containing~$s'$ and let $a$ be the largest index such that $p_a \in A'$ (note that $p_0=s' \in A'$ and $p_q=t' \notin A'$ by definition). It follows that $p_{a+1} \notin A'$ and, since $\{p_{a},p_{a+1}\} \in E'$, there is a $p_a$-$p_{a+1}$ path $P'$ in $G''$ of length at most three.
    As we have $p_a \in A'$ and $p_{a+1} \in V \setminus (A' \cup S')$ and $G[A']$ is a connected component of $G'' - S$, there must be a vertex $x$ of $S$ on $P'$. Since neither $p_a$ nor $p_{a+1}$ is in $S'=N[S]$, it follows that $d_G(p_a,x) \ge 2$ and $d_G(p_{a+1},x) \ge 2$. This contradicts $P'$ having length at most~3.
    
    \laproof{} Let $S'$ be an $s'$-$t'$-separator in $G'$ of size at most $k$. 
    Let $A'$ be the vertex set of the connected component of $G'-S'$ containing $s'$. Consider the set~$S=\{v \in S' \mid d_{G''}(v,A')=2\}$. We claim that $S$ is an $s$-$t$-separator in $G$ and, moreover, that $N[S] \subseteq S'$ and, hence, $|N[S]| \le k$.    
    As to the second part, we have $S \subseteq S'$ by definition. Suppose for contradiction that there was a vertex $u \in N(S)\setminus S'$ that is a neighbor of $v \in S$. Then, since $d_{G''}(v,A')=2$, we have $d_{G''}(u,A') \le 3$, meaning that $u$ has a neighbor in~$A'$ in~$G'$, and, thus $u$ is in $A'$. This implies that $d_{G''}(v,A')=1$, a contradiction. Hence, $|N[S]| \leq k$.
    
    It remains to show that $S$ is an $s$-$t$-separator in~$G$. For this, we prove that $S$ is an $s'$-$t'$-separator in $G''$. Since it contains neither $s$ nor $t$, it follows that it must be also an $s$-$t$-separator in $G$. Assume for contradiction that there is an $s'$-$t'$ path in $G'' -S$. This implies that $d_{(G'' -S)}(t',A')$ is well defined (and finite). Let $q:=d_{(G'' -S)}(t',A')$ and let $P$ be a corresponding shortest path in $G'' -S$. Let us denote $P=p_0, \ldots, p_q$ with $p_q =t'$ and $p_0 \in A'$. If $d_{G''}(t',A')\le 3$, then $t'$ has a neighbor in $A'$ in $G'$, and therefore it is in $A'$ contradicting our assumption that $S'$ is an $s'$-$t'$-separator in $G'$. As $t'=p_q$, we have $q > 3$. Since $d_{G''}(p_0,A')=0$, $d_{G''}(p_q,A')> 3$, and $d_{G''}(p_{i+1},A')\le d_{G''}(p_i,A')+1$ for every $i \in \{0, \ldots q-1\}$, there is an $a$ such that $d_{G''}(p_a,A')=2$. Note that each vertex $v$ with $d_{G''}(v, A') \leq 3$ is either in $A'$ or in $S'$. If $p_a$ is not in $S'$, then $p_a$ is in $A'$, contradicting our assumptions on $P$ and $q$ as $a \ge 2$. Therefore we have $d_{G''}(p_a,A')=2$ and $p_a$ is in $S'$. It follows that $p_a$ is in $S$, a contradiction. \qedhere
\end{proof}

As a minimum $s$-$t$~separator can be computed in polynomial time using standard methods, for example, based on network flows (cf.~e.g.~\cite{KleinbergT06}), \cref{thm:sstspolytime} follows immediately from~\cref{lem:sepsinthirddpower}.

\subsection{Small Secluded \(s\)-\(t\)-Separator}\label{ssec:sssts}

In this subsection we prove hardness and tractability results for the following problem.

\decprob{Small Secluded $s$-$t$-Separator}
{A graph~$G=(V,E)$, two distinct vertices $s,t\in V$, and two integers $k,\ell$.}
{Is there an $s$-$t$~separator $S\subseteq V\setminus \{s,t\}$ such that $|S|\leq k$ and $|N_G(S)|\leq \ell$?}

\noindent We show that, in contrast to \textsc{Secluded
  $s$-$t$-Separator}, the above problem is \NP-hard. Moreover, at the
same time, we show parameterized hardness with respect to $k$ and with
respect to $\ell$.  Later, however, we will show a fixed-parameter
algorithm for the combined parameter~\(k+\ell\).
	
\begin{theorem}%
\label{thm:ssstshard}
\textsc{Small Secluded $s$-$t$-Separator} is \NP-hard and $W[1]$-hard when parameterized by~$k$ or by~$\ell$.
\end{theorem}
In the proof of the theorem,
we reduce from the \textsc{Cutting at Most $k$ Vertices with Terminal}~\citep{FominGK13} problem, which asks, given a graph $G=(V,E)$, a vertex $s\in V$, and two integers $k\geq1$, $\ell\geq0$, whether there is a set $S\subseteq V$ such that $s\in S$, $|S|\leq k$, and $|N_G(S)|\leq \ell$.
\citet{FominGK13} proved that \textsc{Cutting at Most $k$ Vertices with Terminal} is \NP-hard and W[1]-hard when parameterized by~$k$ or by $\ell$.

{
  \begin{proof}
  We give a polynomial-parameter transformation from \textsc{Cutting at Most $k$ Vertices with Terminal} to \textsc{Small Secluded $s$-$t$-Separator}.
  
  \emph{Construction.}
  Let $\I:=(G=(V,E),s,k,\ell)$ be an instance of \textsc{Cutting at Most $k$ Vertices with Terminal}. 
  We construct an instance $\I':=(G',s',t',k',\ell')$ of \textsc{Small Secluded $s$-$t$-Separator} equivalent to $\I$ as follows.
  To obtain $G'$ from $G$ we add to $G$ two vertices $s'$ and $t'$ and two edges $\{s',s\}$ and $\{s,t'\}$.
  Note that $G=G'-\{s',t'\}$.
  We set $k'=k$ and $\ell'=\ell+2$.
  Hence, we ask for an $s'$-$t'$~separator $S\subseteq V(G')\setminus \{s',t'\}$ in~$G'$ of size at most $k'$ with $|N_{G'}(S)|\leq \ell'$. Clearly, the construction can be carried out in polynomial time.
  
  \emph{Correctness.} We show that $\I$ is a yes-instance of~\textsc{Cutting at Most $k$ Vertices with Terminal} if and only if $\I'$ is a yes-instance of \textsc{Small Secluded $s$-$t$-Separator}.
  
  \raproof{}
  Let $\I$ be a yes-instance and let $S\subseteq V(G)$ be a solution to~$\I$, that is, $s\in S$, $|S|\leq k$, and $|N_G(S)|\leq \ell$.
  We claim that $S$ is also a solution to $\I'$.
  Since $s\in S$ and $s'$ and $t'$ are both only adjacent to~$s$, $S$ separates $s'$ from $t'$ in $G'$.
  Moreover, $|S|\leq k=k'$ and, as $N_{G'}(S)=N_G(S) \cup \{s',t'\}$, we have $|N_{G'}(S)|\leq \ell+2=\ell'$.
  Hence, $S'$ is a solution to $\I'$, and $\I'$ is a yes-instance.
  
  \laproof{}
  Let $\I'$ be a yes-instance and let $S'\subseteq V(G')\setminus \{s',t'\}$ be an $s'$-$t'$~separator in $G'$ with $|S'|\leq k'$ and $|N_{G'}(S')|\leq \ell'$.
  We claim that $S'$ is also a solution to $\I$.
  Note that $|S'|\leq k'=k$.
  Since $S'$ is an $s'$-$t'$~separator in $G'$ and $s'$ and $t'$ are both adjacent to~$s$, it follows that~$s\in S'$ and $s',t' \in N_{G'}(S')$.
  Thus, we have $s\in S'$ and $|N_G(S')|=|N_{G'-\{s',t'\}}(S')|=|N_{G'}(S')|-2\leq \ell'-2=\ell$.
  Hence, $S'$ is a solution to~$\I$ and $\I$ is a yes-instance.
  
  Note that, in the reduction, $k'$ and $\ell'$ only depend on $k$ and $\ell$, respectively.
  Since \textsc{Cutting at Most $k$ Vertices with Terminal} parameterized by $k$ or by $\ell$ is W[1]-hard~\citep{FominGK13}, it follows that \textsc{Small Secluded $s$-$t$-Separator} parameterized by~$k$ or by~$\ell$ is W[1]-hard.
   \end{proof}
}
\noindent
Note that the above reduction exploits the fact that we can increase the separator size to decrease the number of vertices in its neighborhood. In fact, most of the vertices declared to be in the separator do not serve to separate the terminals at all. To avoid this idiosyncrasy we suggest to study secluded inclusion-wise minimal separators in future work.

In the following, we prove that \textsc{Small Secluded $s$-$t$-Separator} is FPT when parameterized by~$k+\ell$.
\begin{theorem}
  \label{thm:fptsssts}
  There is a computable function $f$ such that \sssts\ is solvable in
  $f(k + \ell) \cdot m \log n$ time, where $n$ is the number of
  vertices and $m$ the number of edges in the input graph, $k$ is the
  separator size and $\ell$ its open neighborhood size.
\end{theorem}

\noindent
To prove \cref{thm:fptsssts}, we exploit that \sssts{} is efficiently solvable
on graphs of small treewidth:
\begin{lemma}\label[lemma]{ssts:courcelle}
  There is a computable function~$f$ such that \sssts{} can be solved in \(f(k,\ell,w)\cdot n\)~time on graphs of treewidth~\(w\).
\end{lemma}

\begin{proof}
  By Courcelle's theorem~\cite{Cou90}, it is enough to express the
  existence of an \(s\)-\(t\)-separator of cardinality~\(k\) and open
  neighborhood size~\(\ell\) in a formula of the monadic second-order logic of
  graphs such that this formula has a size depending only on~\(k\) and~\(\ell\).\footnote{We refrained from giving a tedious dynamic program
    on a tree-decomposition and opted for Courcelle's theorem instead
    for clarity and because a singly exponential running time is out
    of reach due to the use of the treewidth reduction technique that we
    will apply later.}
  More specifically, using Courcelle's theorem, we want to verify
  whether our graph~\(G=(V,E)\) of treewidth~\(w\) satisfies
  \[
    \exists S\subseteq V:\text{separates}(s,S,t)\wedge \text{card-le}(S, k)\wedge \forall N\subseteq V:(\text{open-nh}(N,S)\Rightarrow \text{card-le}(N,\ell))
  \]
  where the predicate
  \[
    \text{card-le}(S,k)\equiv\forall v_1\forall v_2\dots\forall v_{k+1}:\Bigl(\bigwedge_{i=1}^{k+1}v_i\in S\Bigr)\Rightarrow \Bigl(\bigvee_{i=1}^{k}\bigvee_{j=i+1}^{k+1}v_i=v_j\Bigr)
  \]
  is true if and only if \(|S|\leq k\) and this predicate has a size depending only on~\(k\),
  \[
    \text{open-nh}(N,S)\equiv \forall v\in V: v\in N\Leftrightarrow (v\notin S\wedge\exists u\in S:\text{adj}(u,v))
  \]
  is a constant-size predicate that is true if and only if~\(N=N(S)\), and, finally,
  \begin{align*}
    \text{separates}(s,S,t)&\equiv \exists A \subseteq V\exists B \subseteq V:(\forall v:v \in A \vee v \in B \vee v \in S) \wedge s \in A \wedge t\in B \\
    &\wedge (\forall u \forall v:\text{adj}(u,v) \Rightarrow (u \notin A \vee v \notin B) \wedge (u \notin B \vee v \notin A))
  \end{align*}
  is a constant-size predicate that is true if and only if \(S\)~separates~\(s\) from~\(t\).  Namely it is true if and only if the vertex set $V \setminus S$ can be divided into two sets $A$ and $B$ such that $s\in A$, $t \in B$, and there are no edges between the sets $A$ and $B$.
  \end{proof}

  \noindent
In view of \cref{ssts:courcelle}, to prove \cref{thm:fptsssts} it is enough to reduce \sssts{}
to the case where the input graph has treewidth bounded by some function in~\(k+\ell\).
To this end, we
use the following treewidth reduction
technique.
Let $G$ be a graph and $W \subseteq V(G)$. The
\emph{torso $\torso(G, W)$} is a graph obtained from $G[W]$ by taking
each connected component~$C$ in $G \setminus W$ and making $N(C)$ into
a clique in $G[W]$. The following lemma is implied by
\citeauthor{MOR13}'s Lemma~2.11~\cite{MOR13}:
\begin{lemma}[\citet{MOR13}]\label{lem:twred}
  Let $s, t$ be two vertices of a graph~$G$ and let
  $r \in \mathbb{N}$. Let $W'$ be the union of all inclusion-wise
  minimal $s$-$t$ separators of size at most~$r$. Then, there is an
  $f(r) \cdot (n + m)$-time algorithm that returns a set
  $W \supseteq W' \cup \{s, t\}$ such that $\torso(G, W)$ has treewidth upper
  bounded by some function depending only on~$r$.
\end{lemma}

\noindent\looseness=-1
The algorithm for \cref{thm:fptsssts} now proceeds in three stages.
First, we use random separation~\cite{CCC06}: we randomly color all
vertices red or green.  With sufficiently high probability, all separator vertices
will be colored green and all their neighbors red.  We will henceforth
assume to have such a coloring.  Second, we prepare a graph of low
treewidth containing our solution. This we do by first contracting
each green component in the graph to make our desired separator an
inclusion-wise minimal one. Relying on this, we then use the
treewidth reduction technique (\cref{lem:twred}) to compute a vertex
subset containing the solution which induces a graph of small
treewidth. In this vertex subset, however, neighbors of some $s$-$t$
separators may be missing, possibly introducing false
positives. Furthermore, paths between $s$ and $t$ may be missing,
introducing false positives as well. Hence, we make some modifications
to reintroduce the corresponding information. To the resulting graph,
we then apply \cref{ssts:courcelle}.

\begin{proof}[Proof of \cref{thm:fptsssts}]
    In the following we simultaneously describe the algorithm and gather
  arguments to show that, if there is a solution, then the algorithm
  will accept. For this purpose, assume that the input is a
  yes-instance and, among all separators with at most $k$ vertices and
  open neighborhood of size at most~$\ell$, let $K$ be a separator
  with minimum number of vertices. We will argue in the end that, if
  the input is a no-instance, then the algorithm will reject.

  \smallskip
  \emph{Random separation stage}: We present this stage already in
  derandomized form. We use \citeauthor{NSS95}'s construction of
  universal sets~\cite{NSS95}. An \emph{$(n, d)$-universal set} $\F$
  over some universe $U$ of size $n$ is a family of subsets of $U$
  such that, for each~$A \subseteq U$ of size exactly $d$, the family
  $\{A \cap S \mid S \in \F\}$ contains $A$ and all subsets
  of~$A$. \citet{NSS95} showed that an $(n, d)$-universal set of size
  $2^dd^{O(\log d)}\log n$ can be computed in~$2^dd^{O(\log d)}n \log n$~time.
  In the random separation stage, we compute an
  $(n-2, k + \ell)$-universal set~$\F$ over the vertex set $V \setminus \{s,t\}$ of the
  input graph. We iterate over all the sets $F \in \F$ and in each
  such iteration perform all the algorithm steps described later. Call
  the vertices in~$F$ \emph{green} and the vertices in $V \setminus F$
  \emph{red}. Note that, in one of these iterations, we have that each
  vertex in $K$ is green and each vertex in $N(K)$ is red, by the
  definition of universal sets. Call such an iteration
  \emph{good} for $K$. Our aim is to show that, if we are in a good
  iteration for $K$, we will accept.

  \smallskip \emph{Building a graph of low treewidth}:
  We call an inclusion-wise maximal connected set of green vertices of~\(G\) a
  \emph{green component} of~\(G\).  We
  construct a graph~$G_1$ from~$G$ by contracting each green component
  into one vertex.
We remove all self-loops introduced in the process. In the following,
  we call \emph{component vertices} the vertices in~$G_1$
  corresponding to green components in~$G$. Let $m: V(G) \to V(G_1)$ be the function mapping 
  each red vertex to itself and each green vertex to the corresponding component vertex.

  \looseness=-1
  If we are in a good iteration for~$K$, each green component in~$G$ is
  either contained in~$K$ or disjoint from it.
  We claim that then $K$~induces an inclusion-wise minimal $s$-$t$
  separator~$K'$ in $G_1$, where $K'=m(K)$, that is, we obtain $K'$ from~$K$ by replacing
  each green component in~$K$ by the corresponding component vertex
  in~$G_1$.  Clearly, $K'$ is an $s$-$t$ separator in~\(G_1\). Observe that $K'$ is an
  independent set in~$G_1$. Now, for the sake of a contradiction,
  assume that there is an inclusion-wise minimal separator $K'_1$
  strictly contained in~$K'$. Since $K'$ is an independent set,
  $|N_{G_1}(K'_1)| \leq |N_{G_1}(K')| \leq \ell$. 
  Hence, $m^{-1}(K'_1)$ is an $s$-$t$ separator in~\(G\) with 
  $|m^{-1}(K'_1)| < |K| \le k$ and $|N_G(m^{-1}(K'_1))|=|N_{G_1}(K'_1)| \le \ell$.
  This is a contradiction to the
  fact that~$K$ has minimum number of vertices among 
  $s$-$t$ separators in~\(G\) with at most $k$ vertices and neighborhood size at most~$\ell$.
  Thus, indeed, $K'$~is inclusion-wise minimal.

  We next use treewidth reduction (\cref{lem:twred}) to compute a
  vertex set $W$ containing the vertices~$s,t$ and all inclusion-wise
  minimal $s$-$t$ separators in $G_1$ of size at most~$k$ such that
  $\torso(G_1, W)$ has treewidth upper bounded by a function of~$k$
  alone. Clearly, $K' \subseteq W$. In the following, we construct a
  graph~$G_4$ from $G_1[W]$ so that we can simply compute
  a small secluded \(s\)-\(t\) separator in $G_4$ instead of the input
  graph. (We bypass $\torso(G_1, W)$ and use it only to show that
  $G_4$ has bounded treewidth.) The construction of~$G_4$ is as
  follows.
  \begin{enumerate}[1)]
  \item Initially, let $G_2 = G_1[W]$.\label{step:init}
  \item Denote $B=N_G(m^{-1}(W))$. Add $B$ to $G_2$ and for each $b \in B$ and each $v \in N_G(b)$ if $m(v)$ is in $W$, then make $b$ adjacent to $m(v)$.\label{step:nbs1}
  \item Let $\mathcal{C}$ be the set of connected components of the graph $G_1 - W$. For each $C \in \mathcal{C}$ introduce the new vertex $c_C$ to $G_2$ and make it adjacent to each vertex in $m^{-1}(C) \cap B$. For reference later on, let $D=\{c_C\mid C \in \mathcal{C}\}$.\label{step:comps}
  \item Let us now distinguish the following subsets of $W$.
  Recall that $F$ is the set of green vertices and $m(F)$ is the set of component vertices.
  The set $R= W \setminus m(F)$ is the set of red vertices in $W$.
  The set $H= \{a \in W\mid |m^{-1}(a)| > k\}$ is the set of component vertices in $W$ corresponding to huge green components that cannot take part in the solution.
  Finally, the set $A=\{a \in m(F)\cap W\mid |m^{-1}(a)| \le k\}=W \setminus (R \cup H)$ is the set of component vertices in $W$ allowed in the solution.
  
   Construct $G_3$ from $G_2$ as follows. 
   Replace each component vertex $a$ in $A$ by a clique on vertex set $m^{-1}(a)$ and make each vertex of the clique adjacent to each neighbor of $a$ in $G_2$. Note that, as the neighbors of a green component are all red, none of the neighbors are being replaced in the current step. Let $A'=m^{-1}(A) \subseteq F$ be the set of vertices newly (re-)introduced to $G_3$.\label{step:uncontract}
  \item Finally, to obtain $G_4$ from $G_3$, for each vertex in $R \cup H\cup B \cup D$ introduce $k + \ell + 1$ new
    degree-one neighbors.\label{step:nbs2}
   \end{enumerate}
   We claim that $G_4$
  satisfies the following properties.
  \begin{enumerate}[(i)]
  \item If we are in a good iteration for $K$, then $K \subseteq V(G_4)$, $|N_{G_4}(K)| \leq \ell$, and $K$ is an
    $s$-$t$ separator in~$G_4$.\label{prop:preserve sol}
  \item If there is an $s$-$t$ separator $L$ in $G_4$ of size at most
    $k$ and $|N_{G_4}(L)| \leq \ell$, then there is also such a
    separator in $G$. (Regardless of whether we are in a good
    iteration.)\label{prop:back}
  \item There is a function~$f$ such that the treewidth of $G_4$ is
    bounded from above by $f(k)$.\label{prop:tw}
  \end{enumerate}
  
  Let us prove Property~(\ref{prop:preserve sol}). Assume that we are in a good iteration for $K$. Recall that~$K' \subseteq W$. Hence,~$K'$ is present in~$G_2$ after Step~\ref{step:init} and, clearly, is not modified in Steps~\ref{step:nbs1} and~\ref{step:comps}. Furthermore, each component vertex in~$K'$
  corresponds to a green connected component in~$G$ of size at
  most~$k$. Hence, in Step~\ref{step:uncontract}, $K'$ is replaced
  by~$K$ and thus $K \subseteq V(G_4)$. 
  Furthermore, $N_{G_4}(K) \subseteq W \cup B$. By construction we have $N_{G_4}(K) \cap W = N_{G_3}(K) \cap W = N_{G_2}(K') \cap W=N_{G_1}(K') \cap W$. Thus, for each vertex $r \in N_{G_4}(K) \cap W$ there is a vertex $v \in K$ such that $r$ and $v$ are adjacent in $G$ by the construction of $G_1$. Also, by construction, we have $N_{G_4}(K) \cap B = N_{G_3}(K) \cap B = N_{G_2}(K') \cap B$ and for each $b \in N_{G_4}(K) \cap B$ there is a vertex $v \in K$ such that $b$ and $v$ are adjacent in $G$ by the construction of $G_2$.
  Hence, $N_{G_4}(K) \subseteq N_G(K)$ and $|N_{G_4}(K)| \leq \ell$. 
  
  To prove Property~(\ref{prop:preserve sol}) it remains to show that $K$ is an $s$-$t$ separator in~$G_4$. %
  Suppose for contradiction that there is an $s$-$t$ path $P$ in $G_4- K$.
  Since $G_4$ and $G_3$ only differ in degree-one vertices, $P$ is an $s$-$t$ path also in~$G_3$.
  If $P$ uses vertices in $A'=\{m^{-1}(a) \mid a \in A\}$, then we can replace a part of $P$ between the first and last vertex in $m^{-1}(a)$ by $a$.
  This way we obtain an $s$-$t$ path $P_2$ in $G_2 - K'$. 
  If $P_2$ uses vertices outside~$W$ then let~$Q$ be the set of vertices appearing on $P_2$ between two consecutive vertices $a_1$ and $a_2$ of~$W$ and observe that $Q \subseteq B \cup D$.   
  We claim that we can remove $Q$ from $P_2$ and replace it by a path~$P'$ between~$a_1$ and~$a_2$ in $G_1$ and, hence, obtain a path~$P_1$ in~$G_1$. 
  To see this, note that the vertices $b \in B$ are only connected to vertices in $W \cap N_{G_1}(C)$ and to $c_C$ for a connected component $C$ of $G_1- W$ with~$m(b) \in C$. 
  Furthermore,~$c_C \in D$ is only connected to vertices in $B$. It follows that there is a connected component~$C$ of $G_1- W$ such that $Q \subseteq \{c_C\}\cup \{b \in B\mid m(b) \in C\}$ and $\{a_1,a_2\} \subseteq N_{G_1}(C)$. Hence, there is the claimed path~$P'$ between~$a_1, a_2$ in~$G_1$ and we can replace $Q$ in $P_2$ with~$P'$. 
  Doing this with each part of $P_2$ outside~$W$, we obtain an $s$-$t$~path in $G_1 - K'$, a contradiction.
  Hence, $K$ is an $s$-$t$ separator in $G_4$.

  We now prove Property~(\ref{prop:back}). We may assume that $L$ does not contain
  any vertex that has been introduced in
  Step~\ref{step:nbs2} as they could be removed from~$L$, yielding another separator that fits the definition of~$L$.
  Furthermore, $L$ does not contain any vertex in $R \cup H\cup B \cup D$, because $L$'s neighborhood has size at most~$\ell$.
  Hence, we have $L \subseteq A' \subseteq V(G)$.
  
  Suppose for contradiction that there is an $s$-$t$ path $P$ in $G - L$. As shown above, no (part of a) green component of size more than~$k$ in~$G$ is contained in~$L$. Furthermore, since green components of size at most~$k$ have been contracted in~$G_1$, and then replaced by a clique of vertices with identical neighborhood in~$G_3$, we may assume without loss of generality, that each green component of~$G$ is either completely contained in~$L$ or disjoint from it.
  Hence, contracting green parts of path~$P$ we obtain an $s$-$t$ path $P_1$ in $G_1 - m(L)$.
  For each part of $P_1$ using vertices outside of $W$ we do the following. Obviously, each such part must stay within one connected component $C$ of $G_1 - W$. Let $b_1$ and $b_2$ be the first and last vertex of $C$ on $P_1$. We replace the part of $P_1$ between $b_1$ and $b_2$ by the vertex $c_C$. Doing this with each part of $P_1$ outside of $W$, we obtain an $s$-$t$ path $P_2$ in $G_2 - m(L)$.
  If we replace each vertex $a \in A$ on $P_2$ with an arbitrary vertex of~$m^{-1}(a)$, we obtain an $s$-$t$ path $P_3$ in $G_3 - L$. 
  This path is also present in $G_4 - L$, contradicting $L$ being an $s$-$t$ separator in $G_4$.
  Hence, $L$ is an $s$-$t$ separator
  in~$G$. 

  To prove Property~(\ref{prop:back}) it remains to show that the neighborhood of~$L$ in~$G$ has size at most~$\ell$. We show that $N_{G}(L) \subseteq N_{G_4}(L)$. Recall that $L \subseteq A'$. 
  Let $b$ be a vertex of $N_{G}(L)$ and $v$ one of its neighbors in~$L$. 
  Obviously $m(v)$ is in $A \subseteq W$ and $|m^{-1}(m(v))| \le k$.
  If $b$ is in $F$, then $m(b)=m(v)$ and $b$ is in $N_{G_4}(L)$, since $m^{-1}(m(v))$ is a clique in $G_4$.
  Otherwise, $m(b)=b$. 
  If $b$ is in $W$, then $b$ is a neighbor of $m(v)$ in $G_1[W]$, yielding that $b$ is in $N_{G_4}(L)$.
  Finally, if $b$ is in $G_1 - W$, then $b$ is in $B$, $b$ is adjacent to $m(v)$ in $G_2$, meaning that $b$ is in $N_{G_4}(L)$.
  Hence, $N_{G}(L) \subseteq N_{G_4}(L)$, implying that $|N_{G}(L)| \le \ell$.
  
  We now prove Property~(\ref{prop:tw}). By Lemma~\ref{lem:twred}
  there exists a tree decomposition~$T$ for $\torso(G_1, W)$ and a
  function~$f'$ such that $T$ has width~$f'(k)$. We show how to
  adapt~$T$ into a tree decomposition for~$G_4$
  without increasing its width too much. Clearly, $T$ is also a tree decomposition for $G_1[W]$.
  Recall that the neighborhood of each connected component $C$ of $G_1 - W$ is a
  clique in $\torso(G_1, W)$. Hence, this neighborhood $N_{G_1}(C)$ occurs in one
  bag $Q$ of~$T$. Thus, to incorporate~$c_C$ and its adjacent edges into~$T$, we make a copy~$Q'$ of bag $Q$,
  add~$c_C$ to $Q'$ and make $Q'$ a child of $Q$ in $T$.
  Then for each $b \in B$ adjacent to~$c_C$ in~$G_2$ we add a child bag of $Q'$ with vertex set $Q' \cup \{b\}$.
  Since all neighbors of $b$ in $G_2$ are in $Q'$ and vertices of $D$ only have neighbors in $B$, we obtain a valid tree decomposition  for $G_2$. This increases the width of $T$ by at most~$2$.
  In Step~\ref{step:uncontract} we can replace each component vertex by
  the at most~$k$ vertices in the corresponding green component
  in~$G$, increasing the width of~$T$ by at most a factor
  of~$k$ and obtaining a tree decomposition for $G_3$. 
  Then for each vertex $v$ to receive degree-one neighbors in Step~\ref{step:nbs2} we find an arbitrary bag $Q$ of $T$ containing $v$ and for each of the degree one neighbors of $v$ we create a child bag of $Q$ containing $v$ and the degree one vertex.
  Since each of these bags is of size 2, this does not increase the width of $T$. 
  Making these changes to~$T$, we obtain a tree decomposition for $G_4$ of width at
  most~$f'(k) \cdot k + 2$, as required.

  \smallskip\emph{Conclusion:} Let us show that the above properties
  together with \cref{ssts:courcelle} conclude the proof. In each
  iteration of the random separation phase, we build the graph $G_4$
  and the tree decomposition for it and run the algorithm of
  \cref{ssts:courcelle}.  If any of the iterations accepts, then we
  accept, otherwise we reject.  On one hand, if $K$ is an $s$-$t$
  separator in $G$ of size at most $k$ and $|N_{G}(K)| \leq \ell$,
  then at least one iteration is good for it and the algorithm will
  accept the input by Property~(\ref{prop:preserve sol}). On the other
  hand, if the algorithm accepts, then there is an $s$-$t$ separator
  $L$ in $G_4$ of size at most $k$ and $|N_{G_4}(L)| \leq \ell$ and by
  Property~(\ref{prop:back}) $L$ is such a separator in $G$. Thus, the
  algorithm accepts if and only if we face a yes-instance.
  
  The running time can be upper bounded as follows. Computing the
  universal set~$\F$ takes
  $2^{k + \ell}(k + \ell)^{O(\log (k + \ell))}n \log n$ time. For each
  of the $2^{k + \ell}(k + \ell)^{\log (k + \ell)}\log n$ elements we
  make one iteration, each of which takes the following computation
  time.

  Graph $G_1$ can be constructed in linear time, by first
  finding the green components in linear time. Then, we compute the
  neighborhood for all green components simultaneously by scanning
  over the adjacency lists of each contained vertex, and marking it as
  neighbor for that component. Finally, we scan over all vertices,
  finding the neighborhoods of the corresponding component
  vertices.

  Next, Step~1 can clearly be computed in linear time and
  simultaneously we can find in the same fashion as before the
  connected components in $G_1 - W$, their neighborhoods in $W$, as
  well as the set~$B$ together with, for each vertex in $B$ its
  connected component in~$G_1 - W$. Using this information, we
  can compute Step~2 and~3 in linear time.

  Since each vertex in~$A$ corresponds to a connected component of
  size at most~$k$, Step~4 can be computed in $O(k(n + m))$
  time. Step~5 can be computed in $O((k + \ell)(n + m))$ time, because
  each of the sets~$R, H, B, D$ can be computed in linear time
  alongside the previous computations.

  Hence, for each element of the universal set we have made a constant
  number of $O((k + \ell)(n + m))$-time computations so far. After
  that, we find a tree decomposition for~$G_4$ in $f(k)\cdot (n + m)$
  time, using a fixed-parameter constant-factor approximation
  algorithm~\cite{BodDDFLP16} (recall that $G_4$ has treewidth bounded
  by a function of~$k$ by Property~(\ref{prop:tw})). Finally,
  \cref{ssts:courcelle} again uses fixed-parameter linear time with
  respect to~$k + \ell$. Thus, each of the steps has
  $f(k + \ell) \cdot (n + m)$ running time.
\end{proof}

\noindent
In contrast to \cref{thm:fptsssts}, we show that,  under standard assumptions, the problem does not admit a polynomial-size kernel with respect to the parameter~\(k+\ell\):
\begin{theorem}
  \label{thm:nopksssts}
Unless $\nopk$, \textsc{Small Secluded $s$-$t$-Separator} parameterized by $k+\ell$ does not admit a polynomial kernel.
\end{theorem}

{
  \begin{proof}
  We apply an OR-cross-composition with input problem \textsc{Small Secluded $s$-$t$-Separator}~(SSstS) to \textsc{Small Secluded $s$-$t$-Separator} parameterized by $k+\ell$.

  Let $(\I_q=(G_q,s_q,t_q,k_q,\ell_q))_{q=1,\ldots,p}$ be instances of \textsc{Small Secluded $s$-$t$-Separator}.
  We assume that $\min\{k_i,\ell_i\} \geq 0$ and $\max\{k_i,\ell_i\} \leq |V(G_i)|$ for each $i\in[p]$, because otherwise, we can decide $\I_i$ in polynomial time.
  By virtue of choosing a corresponding polynomial equivalence relation, we also assume that 
  \begin{inparaenum}[(i)]
   \item $0\leq k_i,\ell_i\leq |V(G_i)|$ for all~$i\in[p]$, and
   \item $k_i=k_j$ and $\ell_i=\ell_j$ for all~$i,j\in[p]$.
  \end{inparaenum}
  We OR-cross-compose into one instance $\I=(G,s_1,t_p,k,\ell)$ of \textsc{Small Secluded $s$-$t$-Separator}, with $k:=k_i$ and $\ell:=\ell_i$ for any~$i\in[p]$.

  \emph{Construction}:
  Initially, let $G$ be the disjoint union of $G_1,\ldots,G_p$, that is $G=G_1\cup\cdots\cup G_p$.
  Identify each $t_q$ with $s_{q+1}$ for all $q\in[p-1]$.
  Call the obtained vertex $st_q$, $q\in[p-1]$.
  For each $st_q$, $q\in[p-1]$, add $k+\ell+1$ vertices $x_1^{q},\ldots,x_{k+\ell+1}^q$, and connect them by an edge with $st_q$. 
  We will also refer to $s_1$ as $s$ and as $st_0$ and to $t_p$ as $t$ and as $st_p$.
  This finishes the construction of the instance.
  
  \emph{Correctness}:
  We claim that $\I$ is a \yes-instance if and only if there exists $q\in[p]$ such that $\I_q$ is a \yes-instance.

  \laproof{}
  Let $q\in[p]$ such that $\I_q$ is a yes-instance of SSstS.
  Let $S\subseteq V(G_q)\setminus\{s_q,t_q\}$ be an $s_q$-$t_q$-separator of size at most~$k$ in $G_q$ such that $N_{G_q}(S)\leq \ell$.
  By construction of $G$, for all $V'\subseteq V(G_r)\setminus \{s_r,t_r\}$, $r\in [p]$, it holds that $N_{G}(V')=N_{G_r}(V')$.
  Moreover, since $G$ is obtained by a ``serial'' composition of $\I_1,\ldots,\I_p$, every $s$-$t$ path in $G$ contains $s=st_0,st_1,\ldots,st_{p-1},st_p=t$ in this order.
  Hence, any vertex set $V'\subseteq V(G_r)\setminus \{s_r,t_r\}$ separating $st_{r-1}$ and $st_{r}$ in~$G$, $r\in[p]$, also separates $s$ and $t$ in~$G$.
  Altogether, $S$ is an $s$-$t$-separator in $G$ of size at most~$k$ with $N_{G}(S)=N_{G_q}(S)\leq \ell$.
  Thus, $\I$ is a \yes-instance of~SSstS.

  \raproof{}
  Let $S\subseteq V(G)\setminus\{s,t\}$ be a minimal $s$-$t$~separator (of size at most~$k$) such that $N_G(S)\leq \ell$.
  Observe that $S\cap \{st_1,\ldots,st_{p-1}\}=\emptyset$, since every $st_r$, $r\in[p-1]$, is incident to at least $k+\ell+1$ vertices.
  Moreover, no vertex $x^i_j$, $i\in[p-1]$, $j\in[k+\ell+1]$ is contained in $S$ since $S$ is chosen as minimal and $x^i_j$ is of degree one and hence not participating in any minimal $s$-$t$~separator in~$G$.
  We claim that there exists a $q\in[p]$ with $S\subseteq V(G_q)\setminus\{s_q,t_q\}$.
  Following the argumentation above, since $S$ separates $s$ and $t$, there is at least one $r\in [p]$ such that $S$ separates $st_{r-1}$ and $st_{r}$. 
  Let $q$ be the minimal index such that $S$ separates $st_{q-1}$ and $st_{q}$
  Suppose there is an $r\neq q$ such that $S\cap V(G_r)\setminus\{s_r,t_r\}\neq \emptyset$.
  Since $S$ separates $s$ from $st_{q}$, $S'=S\cap (V(G_q)\setminus\{s_q,t_q\})$ is an $s$-$t$-separator of $G$ of size smaller than~$S$.
  This contradicts the minimality of~$S$.
  Hence, $S\subseteq V(G_q)\setminus\{s_q,t_q\}$.
  Since $S$ separates $st_{q-1}$ and $st_{q}$ in~$G$, it follows that $S$ separates $s_q$ and $t_q$ in $G_q$. 
  Together with $|S|\leq k$ and $N_{G_q}(S)=N_{G}(S)$ implying $|N_{G}(S)|\leq \ell$, it follows that $\I_q$ is a \yes-instance.
  \end{proof}
}

\section{\boldmath$q$-Dominating Set}\label{sec:dom}
In this section, for two constants $p,q\in\mathbb{N}$ with $0 \le p<q$, we study the following problems:%
\decprob{$p$-Secluded $q$-Dominating Set}
	{A graph~$G=(V,E)$ and an integer $k$.}
	{Is there a set $S\subseteq V$ such that $V = N^q_G[S]$ and $|N^p_G[S]|\leq k$?}
\decprob{Small $p$-Secluded $q$-Dominating Set}
	{A graph~$G=(V,E)$ and two integers $k, \ell$.}
	{Is there a set $S\subseteq V$ such that $V = N^q_G[S]$, $|S| \le k$, and $|N^p_G(S)|\leq \ell$?}

\noindent
For $p=0$, the size restrictions in both cases boil down to $|S| \le k$. This is the well-known case of \textsc{$q$-Dominating Set} (also known as $q$-\textsc{Center}) which is \NP-hard and \W{2}-hard with respect to $k$ (see~\citet{LokshtanovMPRS13}, for example). Therefore, for the rest of the section we focus on the case $p > 0$. 
Additionally, by a simple reduction from \textsc{$q$-Dominating Set}, letting $\ell=|V(G)|$, we arrive at the following observation.

\begin{observation}
 For any $0 <p<q$, \textsc{Small $p$-Secluded $q$-Dominating Set} is \W{2}-hard with respect to~$k$.
\end{observation}
Furthermore, let us make the following observation which we will use at multiple occasions in the proofs.
\begin{observation}
\label{obs:cliques}
Let $(G=(V, E), k)$ be a yes-instance of \textsc{$p$-Secluded $q$-Dominating Set} and $S\subseteq V$ be a $q$-dominating set of $G$ with $|N^p_G[S]|\leq k$. If $G$ contains a clique $C$ with $|C| > k$, then for any $v\in V$ and any $c \in C$ such that $d_G(v, c) < p$ we have $v\notin S$.
\end{observation}

  \begin{proof}
  Let $S$ be a $p$-secluded $q$-dominating set of $G$ and $C$ be a clique in $G$ with $|C|>k$. Assume for contradiction that there is a vertex $v \in S$ and a vertex $c \in C$ with $d_G(v, c)<p$. Then we have that $d_G(v, c')\leq p$ for any~$c'\in C$, which implies that $C \subseteq N^p_G[S]$ and hence $|N^p_G[S]|>k$, a contradiction.
  \end{proof}
We now go on to show \NP-hardness and \W{2}-hardness with respect to $k$ for \textsc{$p$-Secluded $q$-Dominating Set}. 
We reduce from the following problem:
\decprob{Set Cover}
{A finite universe $U$, a family $F \subseteq 2^U$, and an integer $k$.}
{Is there a subset $X\subseteq F$ such that $|X| \le k$ and $\bigcup_{x\in X} x=U$?}

\noindent
We write $\bigcup X$ short for $\bigcup_{x\in X} x$. It is known that \textsc{Set Cover} is \NP-complete, \W{2}-hard with respect to $k$, and admits no polynomial kernel with respect to $|F|$, unless \nopk~\citep{dom2014kernelization}.
	
\begin{theorem}%
\label{thm:q-dom-NPh}
 For any $0 <p<q$, \textsc{$p$-Secluded $q$-Dominating Set} is \NP-hard. Moreover, it does not admit a polynomial kernel with respect to $k$, unless \nopk.
\end{theorem}

  \begin{proof}
  We give a polynomial-parameter transformation from \textsc{Set Cover} parameterized by~$|F|$. Let $(U,F,k)$ be an instance of \textsc{Set Cover}. Without loss of generality we assume that $0 \le k <|F|$.  
 
  \emph{Construction.} Let $k'=p+1+|F|\cdot p +k$. We construct the graph~$G$ of a \textsc{$p$-Secluded $q$-Dominating Set} instance~$(G,k')$ as follows.
  We start the construction %
  by taking two vertices $s$ and $r$ and three vertex sets $V_U = \{u \mid u \in U\}$, $V_F = \{v_A \mid A \in F\}$, and $V'_F = \{v'_A \mid A \in F\}$. We connect vertex $r$ with vertex $s$ by a path of length exactly~$q$. For each $A \in F$ we connect vertices~$v_A$ and $r$ by an edge and vertices~$v_A$ and $v'_A$ by a path $t_0^A, t_1^A, \ldots, t_{p}^A$ of length exactly~$p$, where $t_0^A = v_A$ and $t_{p}^A = v'_A$. 
Let us denote $T$ the set of vertices on all these paths (excluding the endpoints).
  All introduced paths are internally disjoint and the internal vertices are all new. We connect a vertex~$v'_A \in V'_F$ with a vertex $u \in V_U$ by an edge if and only if $u \in A$. Furthermore, 
  we introduce a clique $C_{U}$ of size $k'$ and make all its vertices adjacent to each vertex in $V'_{F} \cup V_U$.

   If $q-p \ge 2$, then for each $u \in U$ we create a path $b_0^{u}, b_1^{u}, \ldots, b_{q-p-2}^{u}$ of length exactly $q-p-2$ such that~$b_0^{u}=u$ and the other vertices are new. 
   Let us denote the set of all new vertices introduced in this step $B$.
  Furthermore, in this case, for each $h \in \{0, \ldots, q-p-2\}$ we introduce a clique~$C_h^{u}$ of size $k'$ and make all its vertices adjacent to vertex $b_h^{u}$. %
   Let us denote the set of all vertices introduced in this step $C$.
  If $q-p=1$ we do not introduce any new vertices.
  
  \emph{Correctness}: We show that the original instance of \textsc{Set Cover} is a yes-instance if and only if the constructed instance of \textsc{$p$-Secluded $q$-Dominating Set} is a yes-instance.
  
  \raproof{} 
  Let us next show that if there is an $X\subseteq F$ such that $|X| \le k$ and $\bigcup X=U$, then there is a subset~$S$ of $V(G)$ such that $V(G) = N^q_{G}[S]$, and $|N^p_{G}[S]|\leq k'$. Indeed, take a set $S=\{r\}\cup \{v_A \mid A \in X\}$. 
  It is a routine to check that all the following vertices are in distance at most $q$ to $r$ and, hence, they are in $N^q_{G}[S]$: vertex $s$ and the vertices on the path between $r$ and $s$, vertices in $V_F$ and $V'_F$ and vertices on the paths between them. For each $u \in V_U$ there is $A \in X$ such that $u \in A$ (as $\bigcup X=U$). As $u$ is adjacent to $v'_A$, $u$ is in distance (at most) $p+1 \le q$ to $v_A$ and, thus, it is in $N^q_{G}[S]$. The same holds for vertices of $C_U$. Moreover, for $q-p \ge 2$, each vertex in $B \cup C$ is in distance at most $q-p-1$ to some vertex in $V_U$. Therefore, all vertices in $B \cup C$ are in $N^q_{G}[S]$ and $V(G)=N^q_{G}[S]$.

  It remains to bound $|N^p_{G}[S]|$. The $p$-neighborhood of $r$ is formed by $p$ vertices on the path to $s$, the vertices in $V_F$, and $p-1$ vertices on each path between $V_F$ and $V'_F$, hence the closed $p$-neighborhood is of size exactly $p+1+|F|\cdot p$. For each vertex in $V_F$, its $p$-neighborhood is only formed by vertices already in the neighborhood of $r$, except for the corresponding vertex in $V'_F$. We get that $|N^p_{G}[S]| \le p+1+|F|\cdot p+k$.

  \laproof{} 
  For the other direction, let us assume that there is a subset $S$ of $V(G)$ such that $V(G) = N^q_{G}[S]$ and~$|N^p_{G}[S]|\leq k'$. 

  We first observe that $S \cap (B \cup C \cup T \cup V'_F \cup V_U \cup C_U)= \emptyset$. Notice that for every $b^u_h\in B$ we have that~$\{b^u_h\}\cup C^u_h$ is a clique of size $k'+1$ and hence, by \Cref{obs:cliques}, for $v \in B \cup C$ we have that $v \notin S$. Similarly, for any $v'_A\in V'_F$ we have that $\{v'_A\}\cup C_U$ is a clique of size $k'+1$ and furthermore we have that for each $t\in T$ there is a $v'_A \in V'_F$ such that $d_G(t, v'_A)<p$. It follows by \Cref{obs:cliques} that for every~$v \in T\cup V'_F \cup C_U$ we have that $v\notin S$. For any $u\in V_U$ we have that $\{u\}\cup C_U$ is a clique of size $k'+1$. %
  \Cref{obs:cliques} yields that $S\cap V_U = \emptyset$.
  Hence, $S$ may only contain vertices in $V_F$ and the vertices on the path between $s$ and $r$ (including endpoints).

  Now suppose $r \notin S$. Then $S$ contains one vertex of $V_F$ for each vertex of $V'_F$ (since these are the only vertices to dominate $V'_F$ and each vertex of $V_F$ can dominate one vertex of $V'_F$) and at least one vertex of the path between $s$ and $r$ (as no vertex in $V_F$ dominates $s$). It follows that $|N^p_G[S]| \geq p+1+|F|\cdot(p+1) > k'$, a contradiction. Hence $S$ must contain $r$. Note that for $S=\{r\}$ we have that $|N^p_G[S]| = p+1+|F|\cdot p = k' - k$. Adding a vertex from $V_F$ to $S$ increases the size of the closed $p$-neighborhood of $S$ by one. It follows that $S$ may contain at most $k$ vertices from $V_F$.

  Now, let us denote $X=\{A \mid v_A \in S\}$. We claim that $X$ is a set cover for $U$ of size at most $k$, that is, $|X| \le k$ and $\bigcup X=U$.
  We already know that $|X| \le k$. Suppose $X$ is not a set cover. Then there is a~$u \in U \setminus \bigcup X$. If $q-p=1$, then, since $S$ is a solution, there must be a vertex in $S$ in distance at most~$q$ from~$u$ in~$G$. As $V_U \cap S=\emptyset$, this vertex must be in distance at most $q-1=p$ from $v'_A$ for some~$A \in F$, $u \in A$. If $q-p \ge 2$, then, since $S$ is a solution, there must be a vertex in $S$ in distance at most $q$ from any vertex in $C_{q-p-2}^{u}$ in~$G$. As $(T \cup V_U \cup C_U\cup B \cup C) \cap S=\emptyset$, this vertex must also be in distance at most~$q-(q-p-1)-1=p$ from~$v'_A$ for some $A \in F$, $u \in A$. In both cases it follows that the vertex of $S$ is in $V_F$ and it is the vertex $v_A$. But then $A \in X$ and $u \in \bigcup X$---a contradiction.

  Since the construction can be performed in polynomial time and $k'$ is linear in $|F|$, the results follows.
  \end{proof}

\noindent
In the following, we observe that the parameterized complexity of both problems varies for different choices of $p$ and $q$.

\begin{theorem}
  \label{thm:q-domW2}
 For any $0 < p\le\frac12 q$, \textsc{$p$-Secluded $q$-Dominating Set} is \W{2}-hard with respect to $k$.
\end{theorem}

We prove \cref{thm:q-domW2} via a polynomial parameter transformation from \textsc{Set Cover} parameterized by the size of the solution $k$.
To this end, we apply the following construction.

\begin{constr}
 \label{constr:sds}
 Let $(U,F,k)$ be an instance of \textsc{Set Cover}.
 Let $k'=(k+1) \cdot (2p+1)$ if $3p < q$ and~$k' = 2p+1+k \cdot [p+\frac12(3p-q+1)(q-p)]$ otherwise. We construct a graph $G$ as follows.

  We start the construction (illustrated in \Cref{fig:q-dom}) by creating three vertices $s$, $c$, and $r$ and three vertex sets $V_U = \{u \mid u \in U\}$, $V_F = \{v_A \mid A \in F\}$, and $V'_F = \{v'_A \mid A \in F\}$. 
  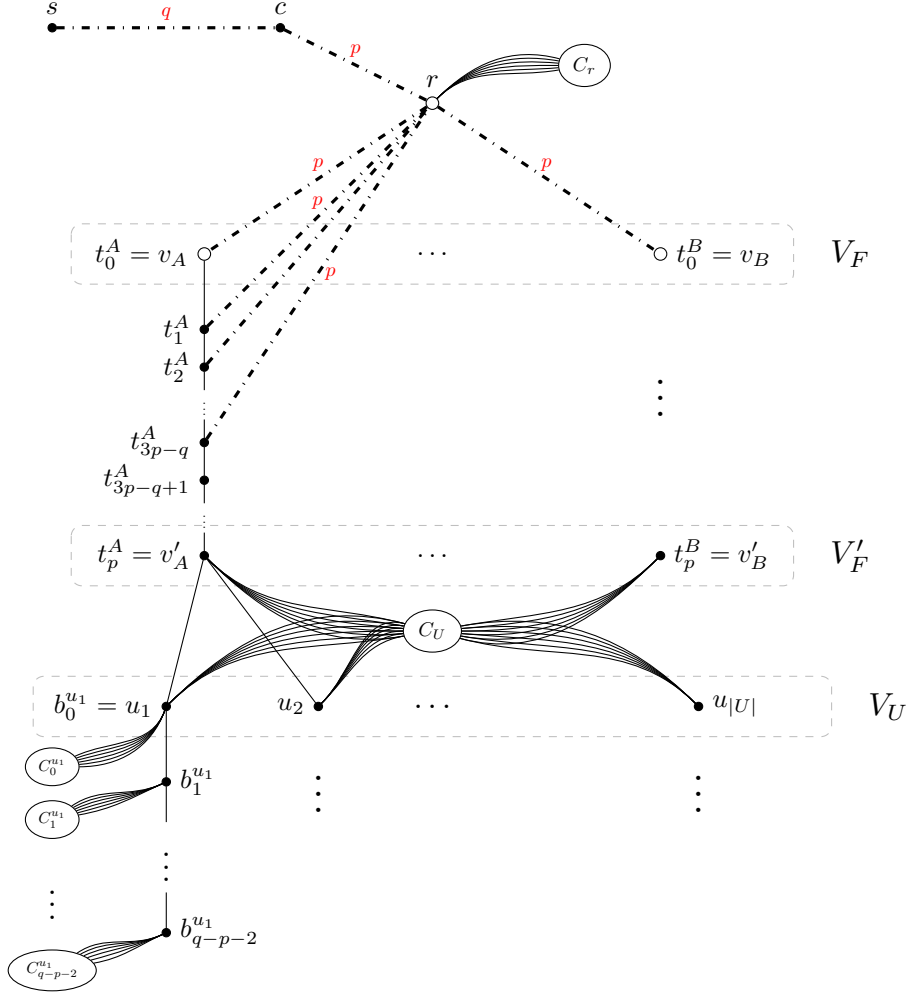
\begin{figure}[t!]
    \centering
    \begin{tikzpicture}
      \usetikzlibrary{calc,shapes}

      \tikzstyle{xnode}=[circle,scale=1/2,draw];
      \tikzstyle{xnode2}=[circle,fill,scale=1/3,draw];
      \tikzstyle{xnodeC}=[scale=0.75];
      \tikzstyle{xnodeC2}=[scale=0.6];
      \tikzstyle{dashdotted}=[dash pattern=on 3pt off 3pt on \the\pgflinewidth off 3pt,line width=1.2]
      \tikzstyle{xpath}=[dashdotted,-];

      \node(s) at (0,0-0.5)[xnode2,label=90:{$s$}]{};
      \node(c) at (3,0-0.5)[xnode2,label=90:{$c$}]{};
      \node(r) at (5,-1-0.5)[xnode,label=90:{$r$}]{};

      \draw[xpath] (s) to node[midway, above,scale=0.75,color=red]{$q$}(c);
      \draw[xpath] (c) to node[midway, above,scale=0.75,color=red]{$p$}(r);

      \node[ellipse,draw,minimum height=10, minimum width=10] (Cr)  at (7,-0.5-0.5)[xnodeC]{$C_r$};
      \foreach \x in {1,2,...,6}{
	      \draw[-] (r) to [in=180-20+5*\x](Cr);
      }

      \draw[rounded corners,dashed,thin,color=lightgray] (0.25,-3.1) rectangle (9.75,-3.9);
      \node[scale=1.2] at (9+1.5,-3.5){$V_F$};

      \node (f1) at (2,-3.5)[xnode,label=180:{$t_0^A=v_A$}]{};
      \node (fld) at (5,-3.5)[]{$\ldots$};
      \node (fn) at (8,-3.5)[xnode,label=0:{$t_0^B=v_B$}]{};

      \draw[xpath] (r) to node[midway, above,scale=0.75,color=red]{$p$}(f1);
      \draw[xpath] (r) to node[midway, above,scale=0.75,color=red]{$p$}(fn);

      \node (t1) at (2,-4-0.5)[xnode2,label=180:{$t_1^A$}]{};
      \node (t2) at (2,-4.5-0.5)[xnode2,label=180:{$t_2^A$}]{};
      \node (tld) at (2,-5-0.5)[scale=1/2]{$\vdots$};
      \node (t3qp) at (2,-5.5-0.5)[xnode2,label=180:{$t_{3p-q}^A$}]{};
      \node (t3qpP1) at (2,-6-0.5)[xnode2,label=180:{$t_{3p-q+1}^A$}]{};
      \node (tld2) at (2,-6.5-0.5)[scale=1/2]{$\vdots$};
      \node (tp) at (2,-7-0.5)[xnode2,label=180:{$t_p^A=v'_{A}$}]{};
      \node (tnld) at (8,-5.25)[scale=1.5]{$\vdots$};

      \node (fld) at (5,-7.5)[]{$\ldots$};

      \draw[-] (f1) to (t1);
      \draw[-] (t1) to (t2);
      \draw[xpath] (t1) to node[midway, above,scale=0.75,color=red]{$p$}(r);
      \draw[-] (t2) to (tld);
      \draw[xpath] (t2) to (r);
      \draw[-] (tld) to (t3qp);
      \draw[-] (t3qpP1) to (t3qp);
      \draw[xpath] (t3qp) to node[midway, right,scale=0.75,color=red]{$p$}(r);
      \draw[-] (t3qpP1) to (tld2);
      \draw[-] (tp) to (tld2);

      \draw[rounded corners,dashed,thin,color=lightgray] (0.25,-7.1) rectangle (9.75,-7.9);
      \node[scale=1.2] at (9+1.5,-7.5){$V'_F$};

      \node[ellipse,draw,minimum height=10, minimum width=10] (Cu)  at (5,-8.5)[xnodeC]{$C_U$};
      \node (tpn) at (8,-7-0.5)[xnode2,label=0:{$t_p^B=v'_{B}$}]{};

      \draw[rounded corners,dashed,thin,color=lightgray] (-.25,-9.1) rectangle (10.25,-9.9);
      \node[scale=1.2] at (9+2,-9.5){$V_U$};

      \node (u1) at (1.5,-9.5)[xnode2,label=180:{$b_0^{u_1}=u_1$}]{};
      \node (u2) at (3.5,-9.5)[xnode2,label=180:{$u_2$}]{};
      \node (uld) at (5,-9.5)[scale=1.2]{$\ldots$};
      \node (um) at (8.5,-9.5)[xnode2,label=0:{$u_{|U|}$}]{};

      \foreach \x in {1,2,...,6}{
	      \draw[-] (tp) to [out=-45,in=180-20+5*\x](Cu);
	      \draw[-] (tpn) to [out=225,in=0-20+5*\x](Cu);
	      \draw[-] (u1) to [out=45,in=180-20+5*\x](Cu);
	      \draw[-] (u2) to [out=45,in=180-20+5*\x](Cu);
	      \draw[-] (um) to [out=-225,in=0-20+5*\x](Cu);
      }

      \draw[-] (tp) to (u1);
      \draw[-] (tp) to (u2);

      \node (bu2vd) at (3.5,-10.5)[scale=1.5]{$\vdots$};
      \node (bumvd) at (8.5,-10.5)[scale=1.5]{$\vdots$};

      \node (bu1) at (1.5,-10.5)[xnode2,label=0:{$b^{u_1}_1$}]{};
      \node (buvd) at (1.5,-11.5)[scale=1.2]{$\vdots$};
      \node (buqp2) at (1.5,-12.5)[xnode2,label=0:{$b^{u_1}_{q-p-2}$}]{};

      \draw[-] (u1) to (bu1);
      \draw[-] (buvd) to (bu1);
      \draw[-] (buvd) to (buqp2);

      \node[ellipse,draw] (ku1) at (0,-10.3)[xnodeC2]{$C^{u_1}_{0}$};
      \node[ellipse,draw] (ku2) at (0,-11)[xnodeC2]{$C^{u_1}_{1}$};
      \node (kuvd) at (0,-12)[scale=1.2]{$\vdots$};
      \node[ellipse,draw] (kuqp2) at (0,-13)[xnodeC2]{$C^{u_1}_{q-p-2}$};

      \foreach \x in {1,2,...,6}{
	      \draw[-] (u1) to [out=-110,in=25-30+5*\x](ku1);
	      \draw[-] (bu1) to [out=-155,in=25-20+5*\x](ku2);
	      \draw[-] (buqp2) to [out=-155,in=25-20+5*\x](kuqp2);
      }

    \end{tikzpicture}
    \caption{%
      Illustration of the construction from \Cref{thm:q-domW2} for $3p \ge q$ and $p>1$. 
      Dash-dotted lines represent paths (where the corresponding length is labeled next to each). 
      The ellipses labeled $C_r$, $C_U$ and $C^u_i$ for some $i$ correspond to cliques of size $k'$.
      If a vertex is connected to such a clique, then there is an edge connecting the vertex with each vertex of the clique (illustrated by multiple lines).
      Vertices represented by empty circles form the set containing a $p$-secluded $q$-dominating set (if any exists).
      }
      \label{fig:q-dom}
  \end{figure}
  We connect the vertices $s$ and $c$ by a path of length exactly $q$ and $c$ with $r$ by a path of length exactly $p$. %
  We connect a vertex $v'_A \in V'_F$ with a vertex $u \in V_U$ by an edge if and only if $u \in A$. 

  We introduce a clique $C_r$ of size $k'$ and make all its vertices adjacent to $r$, and we introduce a clique~$C_{U}$ of size $k'$ and make all its vertices adjacent to each vertex in $V'_{F} \cup V_U$. If $q-p \ge 2$, then for each $u \in V_U$ we further introduce a path $b_0^{u,}, b_1^{u}, \ldots, b_{q-p-2}^{u}$, where $b_0^{u}=u$ and $b_h^{u}$ is a new vertex for each $h \ge 1$. Let $B$ denote the set of all these vertices. Furthermore, in such a case, let us introduce for each $u \in V_U$ a set of cliques $C_0^{u}, C_1^{u}, \ldots, C_{q-p-2}^{u}$, each of size $k'$, and for each $h \in \{0, \ldots, q-p-2\}$ connect every vertex in $C_h^{u}$ by an edge to $b_h^{u}$. Let $C$ denote the union of vertices in these cliques. 

  Then we connect for each $A \in F$ the vertex $v_A \in V_F$ with the vertex $v'_A \in V'_F$ by a path $t_0^{A}, t_1^{A}, \ldots, t_p^{A}$ of length $p$, where $t_0^{A}=v_A$ and $t_p^{A}=v'_A$. Let $T$ denote the set of vertices on all these paths (excluding the endpoints).

  Now let us distinguish two cases. If $3p \ge q$, then we connect $r$ by paths of length $p$ to vertices $t_h^{A}$ for all~$A \in F$ and $h \in \{0, \ldots, 3p-q\}$ (note that $3p-q \le p$ as $2p \le q$). If $3p < q$, then we connect $r$ by paths of length $q-2p$ to vertices  $t_0^{A}=v_A$ for all $A \in F$. Note that in this case $q-2p > p$ and that in both cases the distance between $c$ and any $v'_A$ is exactly $q$. Indeed, in the first case the shortest path contains vertices~$c$,~$r$,~$t_{3p-q}^{A}$, and $t_p^{A}=v'_A$ and the distances are $p$, $p$, and $p-(3p-q)=q-2p$, respectively. In the later case, the shortest path contains vertices $c$, $r$, $v_A$, and $v'_A$ and the distances are $p$, $q-2p$, and $p$, respectively. Let us denote $T'$ the set of vertices introduced in this step. This finishes the construction of the graph $G$.
\end{constr}

Let us now explain the intuition behind~\cref{constr:sds} (cf.~\cref{fig:q-dom}). 
Vertex $c$ can $q$-dominate all vertices in $V_F$, $V_F'$, $C_r$, vertices $r$ and $s$ and the paths connecting these. 
Furthermore, vertex $c$ $q$-dominates the most vertices among all vertices that $q$-dominate $s$ and, thus, we can assume that $c$ is any optimal solution of $(G,k')$. 
The vertices of $G$ that remain to be $q$-dominated are the vertices in $C_U$, $V_U$ and the corresponding vertices in $C$. 
Our construction enforces that the vertices in $V_F$ are the only ones suitable for that and, hence, their selection corresponds to a set cover in the original instance.

\begin{lemma}
 \label{lem:q-domW2}
 Let $(U,F,k)$ be an instance of \textsc{Set Cover} and~$0 < p\le\frac12 q$.
 Let~$(G,k')$ be the instance obtained from~$(U,F,k)$ by applying \cref{constr:sds}.
 Then $(U,F,k)$ is a yes-instance of \textsc{Set Cover} if and only if~$(G,k')$ is a yes-instance of \textsc{$p$-Secluded $q$-Dominating Set}.
\end{lemma}

\begin{proof}
 \raproof{} 
  Let us first show that if there is $X\subseteq F$ such that $|X| \le k$ and $\bigcup X=U$, then there is a subset~$S$ of~$V(G)$ such that $V(G) = N^q_{G}[S]$, and $|N^p_{G}[S]|\leq k'$. Indeed, take a set $S=\{c\}\cup \{v_A \mid A \in X\}$.  
  It is a routine to check that all the following vertices are in distance at most $q$ to $c$ and, hence, they are in $N^q_{G}[S]$: vertices $s$ and $r$ and the vertices on the paths between $c$ and these vertices, vertices in $C_r$, $T'$, $T$, $V_F$, and~$V'_F$. For each $u \in V_U$ there is $A \in X$ such that $u \in A$ (as $\bigcup X=U$). As $u$ is adjacent to $v'_A$, $u$ is in distance (at most) $p+1$ to $v_A$ and, thus, it is in $N^q_{G}[S]$. The same holds for all vertices in $C_U$. Moreover, each vertex in~$B \cup C$ is in distance at most $q-p-1$ to some vertex in $V_U$. Therefore, all vertices in $B \cup C$ are in $N^q_{G}[S]$ and $V(G)=N^q_{G}[S]$.

  It remains to bound $|N^p_{G}[S]|$. The $p$-neighborhood of $c$ is formed by $p$ vertices on the path to~$s$ and~$p$ vertices on the path to $r$ (including $r$ itself), hence it is of size exactly $2p$. For each vertex in $V_F$, its $p$-neighborhood is only formed by vertices in $T \cup T' \cup V'_F \cup \{r\}$. In particular, if $3p\ge q$, then the $p$-neighborhood of $v_A$ contains~$p$ vertices of the form $t_h^{A}$ for $h \in \{1, \ldots, p\}$, $p-1$ vertices on the path from $t_0^{A}$ to $r$ (excluding~$v_A=t_0^{A}$, and~$r$%
  ) and $p-h$ vertices on the path from $t_h^{A}$ to $r$ (excluding $t_h^{A}$) for $h \in \{1, \ldots, 3p-q\}$. This, together gives 
  \begin{align}
  1+p+\sum_{h=0}^{3p-q}(p-h)-1=p+\frac12(3p-q+1)(q-p). \label{eq:neigh-size}
  \end{align}
  If $3p < q$, then the $p$-neighborhood of $v_A$ contains $p$ vertices of the form $t_h^{A}$ for $h \in \{1, \ldots, p\}$ and $p$ vertices on the path from $v_A$ to $r$, that is, $2p$ vertices in total. It follows, that $|N^p_{G}[S]| \le k'$.

  \laproof{} 
  For the other direction, let us assume that there is a subset $S$ of $V(G)$ such that $V(G) = N^q_{G}[S]$, and $|N^p_{G}[S]|\leq k'$. 

  We first observe, that $S \cap (B \cup C \cup T \cup V'_F \cup V_U \cup C_U)= \emptyset$. Notice that for every $b^u_h\in B$ we have that~$\{b^u_h\}\cup C^u_h$ is a clique of size $k'+1$ and hence, by~\Cref{obs:cliques}, for $v \in B \cup C$ we have that $v \notin S$. Similarly, for any $v'_A\in V'_F$ we have that $\{v'_A\}\cup C_U$ is a clique of size $k'+1$ and furthermore we have that for each $t\in T$ there is a $v'_A \in V'_F$ such that $d_G(t, v'_A)<p$. It follows by~\Cref{obs:cliques} that for every~$v \in T\cup V'_F \cup C_U$ we have that $v\notin S$. %
  For any $u\in V_U$ we have that $\{u\}\cup V_U$ is a clique of size $k'+1$. %
  \Cref{obs:cliques} yields that $S\cap V_U = \emptyset$.
  Further observe that in case $3p \ge q$ we have that $d_G(t', r) < p$ for each $t'\in T'$ and, by~\Cref{obs:cliques}, $S \cap T'= \emptyset$.
  It follows that $S$ only contains vertices from $V_F$ and from the path from $s$ to $c$. In case $3p <q$ the set $S$ may also contain vertices from the set $T'$, but notably, all these vertices are in distance at least $p+1$ to any vertex in~$V_U$.

  Now, let us denote $X=\{A \mid v_A \in S\}$. We claim that $X$ is a set cover for $U$ of size at most $k$, that is,~$|X| \le k$ and $\bigcup X=U$.
  Suppose it is not a set cover. Then there is a $u \in U \setminus \bigcup X$. If $q-p =1$, then, since $S$ is a solution, there must be a vertex in $S$ in distance at most $q$ from $u$ in $G$. As $V_U \cap S=\emptyset$, this vertex must be in distance at most $q-1=p$ from $v'_A$ for some $A \in F$, $u \in A$. Similarly, if $q-p \ge 2$, then, since $S$ is a solution, there must be a vertex in $S$ in distance at most $q$ from $v \in C_{q-p-2}^{u}$. As $(B \cup C\cup V_U) \cap S=\emptyset$, this vertex must be in distance at most $q-(q-p)=p$ from $v'_A$ for some $A \in F$, $u \in A$. Since this vertex is not in~$(B \cup C \cup V_U \cup T \cup V'_F)$, and if $3q \ge p$ also not in $T'$, it follows that this vertex is in $V_F$ and it is the vertex $v_A$. But then $A \in X$ and $u \in \bigcup X$, a contradiction.

  Before we show that $|X| \le k$, observe that in $S$, there must be a vertex on the path between $c$ and $s$ (possibly one of the endpoints), this vertex is not in $V_F$, and has a $p$-neighborhood of size at least $p+1$. For~$3p < q$ each vertex of $V_F$ has $p$ vertices of the path to $r$ and to the corresponding vertex of $V'_F$ in its closed $p$-neighborhood. These $2p+1$ vertices are not in $p$-neighborhood of any other vertex in $V_F$. It follows that~$|X| \le |S \cap V_F| \le \frac{k'-(p+1)}{2p+1}=\frac{(k+1) \cdot (2p+1)-(p+1)}{2p+1}<k+1$. Finally, in case $3p \ge q$, the $p$ neighborhood of each vertex in $V_F$ contains $p+\frac12(3p-q+1)(q-p)$ vertices (see Equation~\ref{eq:neigh-size}) which are not in a $p$-neighborhood of any other vertex in $V_F$. It follows that
  \begin{align*}
  |X| &= |S \cap V_F| \\
      &\le \frac{k'-(p+1)}{p+\frac12(3p-q+1)(q-p)}\\
      &=\frac{2p+1+k \cdot [p+\frac12(3p-q+1)(q-p)]-(p+1)}{p+\frac12(3p-q+1)(q-p)}\\
      &=k+\frac{p}{p+\frac12(3p-q+1)(q-p)}\\
      &<k+1.
  \end{align*}
  Thus in all cases $|X| \le k$, finishing the proof of equivalence of the instances.
  Since the construction can be performed in polynomial time and $k'$ is linear in $k$, the result follows.
\end{proof}

As \cref{constr:sds} can be done in polynomial time, \cref{thm:q-domW2} immediately follows from~\cref{lem:q-domW2}.

\noindent
For \textsc{Small $p$-Secluded $q$-Dominating Set}, we remark that we can adapt the reduction for~\Cref{thm:q-dom-NPh}: instead of restricting the size of the closed neighborhood of the $q$-dominating set to at most~$p+1+|F|\cdot p +k$, we restrict the size of the $q$-dominating set to at most~$k+1$ and the size of its open neighborhood to at most~$p+|F|\cdot p$.
Analogously, we can adapt the reduction for~\Cref{thm:q-domW2}. 
This yields the following hardness results.

\begin{corollary}
  \label{cor:q-dom-hard}
For any $0 < p<q$, \textsc{Small $p$-Secluded $q$-Dominating Set} is \NP-hard. Moreover, it does not admit a polynomial kernel with respect to $(k+\ell)$ unless \nopk. For any $0 < p\le\frac12 q$, \textsc{Small $p$-Secluded $q$-Dominating Set} is \W{2}-hard with respect to $(k+\ell)$.
\end{corollary}

\noindent
Now we look at the remaining choices for $p$ and $q$, that is all $p$, $q$ with $p >\frac12 q$. 
In these cases we can show fixed-parameter tractability.
The crucial difference to the case $p\leq \frac12 q$ is the following.
Clearly, in either case, if a solution exists, then the solution is contained in the set~$Y$ of vertices whose $p$-neighborhood has size at most~$k+\ell$.
In the case~$p\leq \frac12 q$ we cannot bound the size of~$Y$ in $k, \ell$ or~$q$: Indeed, in our hardness reduction, the vertex set~$V_F$ contains an in $k, \ell$, and $q$ unbounded number of vertices, but all of these vertices have $p$-neighborhood of size at most $k + \ell$ and are eligible for the solution.
However, in the case~$p >\frac12 q$, we can show that the size of~$Y$ is bounded from above by a function in~$k$,~$\ell$, and~$q$ (otherwise we can answer ``no'').
Roughly, this is because whenever two vertices in~$Y$ $q$-dominate each other, there is a path of length at most~$q < 2p$ connecting them such that all of the vertices on this path have a small neighborhood. This allows us to upper-bound the number of vertices in $Y$ that can be $q$-dominated by one single vertex of~$Y$. Hence, if $Y$ is too large we can reject, and otherwise we can brute force on the set~$Y$ of vertices whose $p$-neighborhood has size at most~$k+\ell$ (which can be found in polynomial time) to find a solution, leading to our fixed-parameter algorithm.

\begin{theorem}%
  \label{thm:q-dom-FPT}
  For any $p >\frac12 q$, \textsc{Small $p$-Secluded $q$-Dominating Set} can be solved in $O(m k^{k+2}(k+\ell)^{qk})$ time and, hence, it is fixed-parameter tractable with respect to $k+\ell$.
\end{theorem}
  \begin{proof}
  Consider a solution~\(S\) for an instance $(G,k,\ell)$ of \textsc{Small $p$-Secluded $q$-Dominating Set}.
  If $x \in S$, then $|N^{p}[x]| \le k+\ell$, since $|S|\le k$ and $|N^p(S)| \le \ell$. Moreover, $|N^{p}[x]| \le k+\ell$ implies $|N[v]| \le k+\ell$ for every $v \in N^{p-1}[x]$. It follows that, if $|N^{p}[y]| \le k+ \ell$ and $y \notin S$, then for each $x \in N^q[y] \cap S$ every vertex on every $x$-$y$ path of length at most $2p-1 \ge q$ has degree at most $k+\ell -1$, since each such vertex
  has distance at most~$p-1$ to~$x$ or~$y$. 
  We point out that this property only holds if~$p>q/2$, making this case different to the case $0< p\leq q/2$ (see~\cref{cor:q-dom-hard}).
  
  If $k+\ell=1$, then either $G$ has at most one vertex or $(G,k,\ell)$ is a no-instance. Hence, in the following, we assume $k+\ell \ge 2$.
  We call vertices $u$ and $v$ \emph{linked}, if there is a path of length at most $q$ between $u$ and $v$ in $G$ such that the degree of every vertex on the path is at most $k+\ell-1$. 
  Let $B[u]=\{v \mid u \text{ and } v \text{ are linked}\}$. 
  We claim that $|B[v]| \le (k+\ell)^q$ for any~$v$.

  To prove the claim, let us denote $B^i[v]$ the set of vertices $u$ such that there is a path of length at most~$i$ between $u$ and $v$ in $G$ such that the degree of every vertex on the path is at most $k+\ell-1$. Obviously,~$B^q[v]=B[v]$. We prove by induction, that $|B^i[v]| \le (k+\ell)^i$ for every $i \in \mathbb{Z}^+_0$. The claim then follows. For $i=0$ we have $B^i[v]=\{v\}$ and $|B^i[v]|=1 \le (k+\ell)^0=1$, showing the basic step. For $i \ge 1$ and~$u \in B^i[v] \setminus B^{i-1}[v]$, let $p_0, p_1, \ldots, p_{i}$ be the $u$-$v$ path showing that $u \in B^i[v]$, i.e., $p_0=u$, $p_{i}=v$, and for every $j \in \{0, \ldots, i\}$ we have $\deg p_j \le k+\ell-1$. Then $p_1 \in B^{i-1}[v]$ and $u$ is a neighbor of $p_1$. Since the vertices in $B^{i-1}[v]$ have in total at most $(k+\ell-1)\cdot|B^{i-1}[v]|$ neighbors and $|B^{i-1}[v]| \le (k+\ell)^{i-1}$ by induction hypothesis, we have $|B^i[v]| \le |B^{i-1}[v]| +(k+\ell-1)\cdot|B^{i-1}[v]| = (k+\ell)\cdot|B^{i-1}[v]| \le (k+\ell) \cdot (k+\ell)^{i-1} =(k+\ell)^i$. This gives the induction step and finishes the proof of the claim.

  Let $Y = \{y \mid |N^{p}[y]| \le k+\ell\}$. Obviously, we have $S \subseteq Y$, since $|N^p[S]| \le k+\ell$. If $y \in Y \setminus S$, then there is $x \in S$ such that $x$ and $y$ are linked. It follows that $y \in B[x]$ and, thus, $Y \subseteq \bigcup_{x \in S}B[x]$. Hence,~$|Y| \le k \cdot (k+\ell)^q \le (k+\ell)^{q+1}$. 
  
  This suggests the following algorithm for \textsc{Small $p$-Secluded $q$-Dominating Set}: 
  Find the set $Y$. If~$|Y| > k \cdot (k+\ell)^q$, then answer ``no''. 
  Otherwise, for each $k' \le k$ and each size-$k'$ subset~$S'$ of~$Y$, check whether $S'$~is a $p$-secluded $q$-dominating set in~$G$.  
  If any such set is found, return it.   
  Otherwise, answer ``no''. 
  Since $S \subseteq Y$, this check is exhaustive.%
  
  As to the running time, the set~$Y$ can be determined in $O(n (k+\ell))$ time by running a BFS from each vertex and stopping it after it discovers $k+\ell$ vertices or all vertices in distance at most $p$, whichever occurs earlier. Then, there are $k \cdot \binom{k \cdot (k+\ell)^q}{k} \le k^{k+1}(k+\ell)^{qk}$ candidate subsets of $Y$. 
  For each such set~$S'$ we can check whether it is a $p$-secluded $q$-dominating set in $G$ by running a BFS from each vertex of $S'$ and marking the vertices which are in distance at most $p$ and at most $q$, respectively. %
  This takes $O(m k)$ time. Hence, in total, the algorithm runs in $O(m k^{k+2}(k+\ell)^{qk})$ time. 
  \end{proof}

\noindent
By~\Cref{obs:reducibility}, the previous result 
transfers to \textsc{$p$-Secluded $q$-Dominating Set} parameterized by $k$. 

\begin{corollary}
  \label{cor:q-dom-FPT}
 For any $p >\frac12 q$, \textsc{$p$-Secluded $q$-Dominating Set} %
 is fixed-parameter tractable with respect to $k$.
\end{corollary}

\section{\boldmath\(\F\)-free Vertex Deletion}\label{sec:ffvd}
\noindent In this section, we study the \ffvd{} ($\F$-FVD) problem for families~\(\F\) of graphs with at most a constant number~\(c\) of vertices, that is, the problem of destroying all induced subgraphs isomorphic to graphs in~\(\F\) by at most \(k\)~vertex deletions.  The problem can, in particular, model various graph clustering tasks \cite{BMN12,HKMN10}, where the secluded variants can be naturally interpreted as removing a small set of outliers that are weakly connected to the clusters.

\subsection{Secluded \(\F\)-free Vertex Deletion}
In this section, we prove a polynomial-size problem kernel for \sffvd{}, where \(\F\)~is a family of graphs with at most a constant number~\(c\) of vertices:
\decprob{Secluded $\F$-free Vertex Deletion}
	{A graph~$G=(V,E)$ and an integer $k$.}
	{Is there a set~$S\subseteq V$ such that $G-S$ is $\F$-free and $|N_G[S]|\leq k$?}

\noindent Henceforth, we call a set $S\subseteq V$ such that \(G-S\)~is \(\F\)-free an \emph{\(\F\)-free vertex deletion set}.

Note that \sffvd{} can be polynomial-time solvable for some families~\(\F\) for which $\F$-FVD~is NP-hard: \textsc{Vertex Cover} (where \(\F\)~contains only the graph consisting of a single edge) is NP-hard,  yet any vertex cover~\(S\) satisfies \(N[S]=V\).  Therefore, an instance to \textsc{Secluded Vertex Cover} is a yes-instance if and only if~\(k\geq n\). In general, however, one can show that \textsc{Secluded $\F$-Free Vertex Deletion} is NP-complete for every family~$\F$ that includes only graphs of minimum vertex degree two (\cref{thm:sffcdnphard}). We mention in passing that, from this peculiar difference of the complexity of \textsc{Vertex Cover} and \textsc{Secluded Vertex Cover}, it would be
interesting to find properties of $\F$ which govern whether \sffvd{} is NP-hard or polynomial-time solvable along the lines of the well-known dichotomy results~\citep{FGMN11,LY80}.

  \begin{theorem}\label{thm:sffcdnphard}
    For each family~$\F$ containing only graphs of minimum
    vertex degree two, \textsc{Secluded $\F$-Free
      Vertex Deletion} is NP-complete.
  \end{theorem}

{
\begin{proof}
    We reduce from \textsc{$\F$-free Vertex Deletion}, which is
    NP-complete for all $\F$ that contain only graphs of minimum
    vertex degree two~\citep{LY80}. Given an instance~$(G, s)$ of
    \textsc{$\F$-free Vertex Deletion} where $G$ contains $n$
    vertices, we add $n+ 1$ new degree-one neighbors to each vertex in
    $G$. In this way, we obtain an instance $(G', k)$ of
    \textsc{Secluded $\F$-free Vertex Deletion} by setting
    $k = s\cdot (n + 1) + n$.
    
    Clearly, each {\(\F\)-free
      vertex deletion set} $S$ of size at most~$s$ for $G$ is a {\(\F\)-free
      vertex deletion set} for $G'$ as each graph in $\F$ contains no
    degree-one vertices. Furthermore, clearly,
    $|N_{G'}[S]| \leq s \cdot (n + 1) + n = k$.

    In the other direction, for each {\(\F\)-free
      vertex deletion set}~$S$ for $G'$ we may assume that no
    degree-one vertex is contained in~$S$, hence, $S$ is a
    {\(\F\)-free
      vertex deletion set} for~$G$. Furthermore, as each vertex
    in~$V(G)$ incurs at least $n+1$ vertices in the closed
    neighborhood~$N_{G'}(S)$, if this neighborhood has size at
    most~$k$, then there are at most~$s$ vertices in~$S$.
  \end{proof}
}

\noindent
It is easy to see that \sffvd{} is fixed-parameter tractable. More specifically, it is solvable in \(c^k\cdot\poly(n)\)~time: simply enumerate all inclusion-minimal \(\F\)-free vertex deletion sets~\(S\) of size at most~\(k\) using the standard search tree algorithm described by~\citet{Cai96} and check~\(|N[S]|\leq k\) for each of them.  This works because, for any \(\F\)-free vertex deletion set~\(S\) with \(|N[S]|\leq k\), we can assume that \(S\)~is an inclusion-minimal \(\F\)-free vertex deletion set since \(|N[S']|\leq |N[S]|\) for every~\(S'\subsetneq S\).  

We complement this observation of fixed-parameter tractability by the following kernelization result.

\begin{theorem}\label{thm:sffvdpolyker}
\textsc{Secluded $\F$-free Vertex Deletion} has a problem kernel comprising \(O(k^{c+1})\)~vertices, where \(c\)~is the maximum number of vertices in any graph of~\(\F\).
\end{theorem}

\noindent Our proof of \cref{thm:sffvdpolyker} exploits \emph{expressive} kernelization algorithms for \hs d~\citep{Bev14,Bev14b,FK15}, %
which preserve inclusion-minimal solutions and that return subgraphs of the input hypergraph as kernels:
Herein, given a hypergraph~\(H=(U,\C)\) with \(|C|\leq d\) for each~\(C\in \C\), and an integer~\(k\), \hs d asks  whether there is a \emph{hitting set}~\(S\subseteq U\) with \(|S|\leq k\), that is, \(C\cap S\ne\emptyset\) for each \(C\in\C\).
\noindent Our kernelization for \sffvd{} is based on transforming the input instance~\((G,k)\) to a \hs d instance~\((H,k)\), computing an expressive \hs d problem kernel~\((H',k)\), and outputting a \sffvd{} instance~\((G',k)\), where \(G'\)~is the graph induced by the vertices remaining in~\(H'\) together with at most \(k+1\)~additional neighbors for each vertex in~\(G\).

\begin{definition}\label{def:kernels}
  Let \((G=(V,E),k)\) be an instance of \sffvd{}.  For a vertex~\(v\in V\), let \(N_j(v)\subseteq N_G(v)\)~be a set of \(j\)~arbitrary neighbors of~\(v\), or \(N_j(v):=N_G(v)\) if \(v\)~has degree less than~\(j\).  For a subset~\(S\subseteq V\), let \(N_j(S):=\bigcup_{v\in S}N_j(v)\).  Moreover, let
  \begin{itemize}[\(H'={}\)]
  \item[\(c:={}\)]\(\max_{F\in\F}|V(F)|\) be the maximum number of vertices in any graph in~\(\F\),
  \item[\(H={}\)]\((U,\C)\) be the hypergraph with \(U:=V\) and \(\C:=\{S\subseteq V\mid G[S]\in\F\}\),
  \item[\(H'={}\)]\((U',\C')\) be a subgraph of~\(H\) with \(|U'|\in O(k^c)\) such that each set~\(S \subseteq U\) with \(|S|\leq k\)~is an inclusion-minimal hitting set for~\(H\) if and only if it is for~\(H'\), and
  \item[\(G'={}\)]\((V',E')\) be the subgraph of~\(G\) induced by~\(U'\cup N_{k+1}(U')\).
  \end{itemize}
\end{definition}

\noindent 
To prove \cref{thm:sffvdpolyker}, we show that \((G',k)\)~is a problem kernel for the input instance~\((G,k)\).  The subgraph~\(H'\) exists and is computable in linear time from~\(H\)~\citep{Bev14b,FK15}.  Moreover, for constant~\(c\), one can compute \(H\) from~\(G\) and \(G'\)~from~\(H'\) in polynomial time.  It is obvious that the number of vertices of~\(G'\) is \(O(k^{c+1})\).  Hence, it remains to show that \((G',k)\)~is a yes-instance if and only if~\((G,k)\)~is.  This is achieved by the following two lemmas.

\begin{lemma}\label[lemma]{lem:eqneigh}%
  For any \(S\subseteq U'\) with \(|N_{G'}[S]|\leq k\), it holds that \(N_G[S]=N_{G'}[S]\).
\end{lemma}

{
  \begin{proof}
    Since \(S\subseteq U'\subseteq V'\cap V\) and since \(G'\)~is a subgraph of~\(G\), it is clear that \(N_G[S]\supseteq N_{G'}[S]\).  For the opposite direction, observe that each \(v\in S\)~has degree at most~\(k\) in~\(G'\).  Thus, \(v\)~has degree at most~\(k\) in~\(G\) since, otherwise, \(k+1\)~of its neighbors would be in~\(G'\) by construction.  Thus, \(N_{G'}(v)\supseteq N_{k+1}(v)=N_{G}(v)\) for all~\(v\in S\) and, thus, \(N_{G'}[S]\supseteq N_G[S]\).
  \end{proof}
}

\begin{lemma}\label[lemma]{lem:sffvdgprime}
  Graph~\(G\) allows for an \(\F\)-free vertex deletion set~\(S\) with \(|N_G[S]|\leq k\) if and only if \(G'\) allows for an \(\F\)-free vertex deletion set~\(S\) with \(|N_{G'}[S]|\leq k\).
\end{lemma}

{
  \begin{proof}
    Let \(S\)~be an inclusion-minimal \(\F\)-free vertex deletion set with \(|N_G[S]|\leq k\) for~\(G\).  Then \(S\)~is an inclusion-minimal hitting set for~\(H\) and, by construction, also for \(H'\).  Thus, \(S\)~consists only of vertices of~\(G'\).
  Since \(G'\)~is an induced subgraph of~\(G\), it holds that \(G'-S\) is an induced subgraph of~\(G-S\), which is \(\F\)-free.  Thus, \(G'-S\) is also \(\F\)-free and \(S\)~is an \(\F\)-free vertex deletion set for~\(G'\).  Moreover, \(|N_{G'}[S]|\leq |N_{G}[S]|\leq k\) follows since \(G'\)~is a subgraph of~\(G\).

  Now, let \(S\)~be an \(\F\)-free vertex deletion set with \(|N_{G'}[S]|\leq k\) for~\(G'\).  Then \(S\cap U'\)~is a hitting set for~\(H'\): if there was a set \(C\in\C'\) with \(C\cap S=\emptyset\), then, by construction of~\(H'\) and~\(G'\), \(G'[C]=G[C]\in\F\) would remain a forbidden induced subgraph in~\(G'-S\). Thus, \(S\)~contains an inclusion-minimal hitting set~\(S'\subseteq U'\) for~\(H'\).  Since \(|S'|\leq k\), it is, by construction of~\(H'\), also a hitting set for~\(H\).  Now, by construction of \(H\) from~\(G\), \(S'\) is a \(\F\)-free vertex deletion set for~\(G\).  Finally, since \(|N_{G'}[S']|\leq|N_{G'}[S]|\leq k\), we also have \(|N_{G}[S']|=|N_{G'}[S']|\leq k\) by \cref{lem:eqneigh}.
  \end{proof}
}

\subsection{Small Secluded \(\F\)-free Vertex Deletion}
In this subsection, we present a fixed-parameter algorithm for the following problem parameterized by~\(\ell+k\).
\decprob{\ssffvd}
	{A graph~$G=(V,E)$ and two integers $k, \ell$.}
	{Is there a subset $S\subseteq V$ such that \(G-S\)~is \(\F\)-free, \(|S|\leq k\), and $|N_G(S)|\leq \ell$?}

\noindent As before, we call a set $S\subseteq V$ such that \(G-S\)~is \(\F\)-free an \emph{\(\F\)-free vertex deletion set}.

In the previous section, we discussed a simple search tree algorithm for \sffvd{} that was based on the fact that we could assume that our solution is an inclusion-minimal \(\F\)-free vertex deletion set.  However, an \(\F\)-free vertex deletion set~\(S\) with \(|S|\leq k\) and \(|N_G(S)|\leq\ell\)~is not necessarily inclusion-minimal: some vertices may have been added to~\(S\) just in order to shrink its open neighborhood.  However, the following simple lemma limits the number of possible candidate vertices that can be used to enlarge~\(S\) in order to shrink~\(N(S)\), which we will use in a branching algorithm.

\begin{lemma}\label[lemma]{lem:nebobo}
  Let \(S\)~be an \(\F\)-free vertex deletion set and \(S'\supseteq S\) such that \(|S'|\leq k\) and \(|N_G(S')|\leq\ell\), then
\(
|N_G(S)|\leq \ell+k
\).
\end{lemma}
\begin{proof}
\(|N_G(S)|=|N_G[S]\setminus S|\leq|N_G[S']\setminus S|\leq|N_G[S']|\leq|N_G(S')\cup S'|\leq\ell+k\).
\end{proof}

\begin{theorem}\label[theorem]{thm:ssffvd}%
  \ssffvd{} can be solved in \(\max\{c,k+\ell\}^{k}\cdot\poly(n)\)-time,  where \(c\)~is the maximum number of vertices in any graph of~\(\F\).
\end{theorem}

{
  \begin{proof}
    First, enumerate all inclusion-minimal \(\F\)-free vertex deletion sets~\(S\) with \(|S|\leq k\).  This is possible in \(c^k\cdot\poly(n)\)~time using the generic search tree algorithm described by~\citet{Cai96}.  For each \(k'\leq k\), this search tree algorithm generates at most \(c^{k'}\) sets of size~\(k'\).  For each enumerated set~\(S\) of \(k'\)~elements, do the following:
  \begin{enumerate}
  \item If \(|N_G(S)|\leq\ell\), then output~\(S\) as our solution.
  \item If \(|N_G(S)|>\ell+k\), then \(S\) cannot be part of a solution~\(S'\) with \(N_G(S')\leq\ell\) by \cref{lem:nebobo}, we proceed with the next set.
  \item Otherwise, initiate a recursive branching: recursively branch into at most \(\ell+k\)~possibilities of adding a vertex from~\(N_G(S)\) to~\(S\) as long as \(|S|\leq k\).
  \end{enumerate}
  The recursive branching initiated at step~3 stops at depth~\(k-k'\) since, after adding \(k-k'\)~vertices to~\(S\), one obtains a set of size~\(k\).  Hence, the total running time of our algorithm is
  \[  \poly(n)\cdot\sum_{k'=1}^kc^{k'}(\ell+k)^{k-k'}=  \poly(n)\cdot\sum_{k'=1}^k\max\{c,\ell+k\}^k=\poly(n)\cdot\max\{c,\ell+k\}^k.\qquad\qedhere
  \]
  \end{proof}
}

\noindent Given \cref{thm:ssffvd}, a natural question is whether the problem allows for a polynomial kernel.

\section{Feedback Vertex Set}\label{sec:fvs}

In this section, we study secluded versions of the \textsc{Feedback Vertex Set (FVS)} problem, which asks, given a graph~$G$ and an integer $k$, whether there is a set~$W\subseteq V(G)$, $|W|\leq k$, such that $G-W$ is cycle-free.

\subsection{Secluded Feedback Vertex Set}\label{ssec:sfvs}

\newcommand{\cc}{{\mathcal{C}}}
\newcommand{\lcac}{{\rm{lcac}}}

We show in this subsection that the problem below is NP-hard and admits a polynomial~kernel.

\decprob{Secluded Feedback Vertex Set (SFVS)}
	{A graph~$G=(V,E)$ and an integer $k$.}
	{Is there a set~$S\subseteq V$ such that $G-S$ is cycle-free and $|N_G[S]|\leq k$?}

\begin{theorem}
  \label{thm:sfvsisnphard}
\textsc{Secluded Feedback Vertex Set} is \NP-hard.
\end{theorem}
The proof is by a reduction from the \textsc{FVS} problem and works by attaching to each vertex in the original graph a large set of new degree-one neighbors.
{
\begin{proof}
    We provide a polynomial time many-one reduction from \textsc{Feedback Vertex Set}.
    Let $(G=(V,E),k)$ be an instance of \textsc{Feedback Vertex Set}. 
    We construct an equivalent instance $(G'=(V',E'),k')$ of SFVS as follows.
    To obtain $G'$, for each vertex $v\in V(G)$ add $n^2$ vertices and connect them to $v$.
    Observe that the \emph{added vertices} have degree one and thus are never part of a cycle in $G'$.
    Further, set $k'=k\cdot (n^2+n)$.
    We claim that $(G,k)$ is a yes-instance of FVS if and only $(G',k')$ is a yes-instance of SFVS.
    
    \raproof{}
    Let $S\subseteq V(G)$ be a feedback vertex set in~$G$. 
    Then the corresponding vertices in $G'$ form a feedback vertex set in~$G'$.
    Moreover, we have $k$ vertices, each having at most~$n^2+n$ neighbors.
    Thus, $|N_{G'}[S]|\leq k\cdot (n^2+n)=k'$.
    It follows that $(G',k')$ is a yes-instance of SFVS.
    
    \laproof{}
    Conversely, let $S$ be a minimal solution to $(G',k')$, that is, $S$ is a feedback vertex set in~$G'$ such that $|N_{G'}[S]|\leq k'$ and $S\setminus \{v\}$ is not a feedback vertex set in $G'$ for every $v\in S$.
    By minimality of $S$, and since the added vertices do not appear in any cycle in $G'$, $S$ does not contain any of the added vertices.
    Hence $S \subseteq V$ and, thus $|S|\le k$ as each vertex in $V$ has at least $n^2$ private neighbors in $G'$.
    Thus, since $S$ forms a feedback vertex set in $G'$, $S$ also forms a feedback vertex set in~$G$.
    It follows that $(G,k)$ is a yes-instance of FVS.
\end{proof}%
}%
On the positive side, \textsc{SFVS} remains fixed-parameter tractable with respect to~$k$:
\begin{theorem}
  \label{thm:sfvspolyker}
\textsc{Secluded Feedback Vertex Set} admits a kernel with $O(k^5)$ vertices.
\end{theorem}

\noindent
In the remainder of this section, we describe the data reduction rules that yield the polynomial-size problem kernel. 
The reduction rules are inspired by the kernelization algorithm for the \textsc{Tree Deletion Set} problem given by~\citet{GiannopoulouLSS16}.

We start by introducing the following notation.
A \emph{2-core}~\citep{SEIDMAN1983269} of a graph $G$ is a maximum subgraph $H$ of $G$ such that, for each $v \in V(H)$, we have $\deg_H(v) \ge 2$.
Note that a 2-core~$H$ of a given graph~$G$ is unique and can be found in polynomial time~\citep{SEIDMAN1983269}.
If $H$~is a 2-core of~$G$, then we use $\deg_{H|0}(v)$ to denote $\deg_H(v) $ if $v \in V(H)$ and $\deg_{H|0}(v)=0$ if $v \notin V(H)$.

\begin{observation}%
\label{obs:2core}
  Let $G$ be a graph, $H$ its 2-core, and $C$ a connected component of~$G - V(H)$. Then $|N(C)\cap V(H)|\leq 1$ and $|N(H)\cap V(C)|\leq 1$.
\end{observation}

  \begin{proof}
    We only show the first statement.  The second statement follows analogously.
  Towards a contradiction, assume that $|N(C)\cap V(H)|\geq 2$.
  Then, there are vertices~$x,y\in V(H)$ with~$x\neq y$ such that $x$ and $y$ have neighbors~$a, b\in V(C)$.
  If $a=b$, then $G'=G[V(H)\cup \{a\}]$ is a subgraph of $G$ such that $\deg_{G'} (v)\ge 2$ for every $v \in V(G')$, contradicting the choice of~$H$ as the 2-core of $G$.
  If $a\neq b$, then, since $C$ is connected, there is a path~$P_C$ in $C$ connecting $a$ and $b$. 
  Thus, $G'=G[V(H)\cup V(P_C)]$ is a subgraph of~$G$ such that $\deg_{G'} (v)\ge 2$ for every~$v \in V(G')$, again contradicting the choice of~$H$  as the 2-core of $G$.
  \end{proof}

\noindent Note that only the vertices in the 2-core are involved in cycles of $G$. 
However, the vertices outside the 2-core can influence the size of the closed neighborhood of the feedback vertex set.
Next, we apply the following reduction rules to our input instance with~$G$ given its 2-core~$H$.

We say that a feedback vertex set~$F$ in~$G$ is secluded if $|N[F]| \le k$.
Further, we say that a secluded feedback vertex set~$F$ in~$G$ is \emph{minimal}, if $F\setminus \{v\}$ is not a secluded feedback vertex set in~$G$ for all~$v\in F$.

\begin{rrule}\label{rr:outside}
If~$\deg_{H|0}(v)=0$ for every $v \in N[u]$, then delete~$u$.
\end{rrule}

\begin{proof}[Proof of correctness]
  Let $F$ be a minimal secluded feedback vertex set in $G$. 
  Since $\deg_{H|0}(v)=0$ for all $v\in N[u]$, none of them is involved in a cycle.
  Hence, $N[u]\cap F=\emptyset$.
  In particular, it follows from $N(u)\cap F=\emptyset$ that $u\not\in N[F]$.
  Hence, $F$ is a secluded feedback vertex set in $G-\{u\}$ as well.

  Conversely, let $F$ be a minimal secluded feedback vertex set in~$G_u:=G-\{u\}$. 
  We have to show that $F$ is a secluded feedback vertex set in~$G$ as well.
  First observe that since $\deg_{H|0}(v)=0$ for all $v\in N_G[u]$, $H$ is also the 2-core of $G_u$.
  As only vertices in $H$ participate in cycles of $G_u$ and $F$ is chosen as minimal, none of the vertices $N_G(u)\subseteq V(G_u)$ is contained in $F$.
  If follows that $|N_{G}[F]|=|N_{G_u}[F]|\leq k$, and thus $F$ is a secluded feedback vertex set in~$G$ as well.
\end{proof}

\noindent
Note that, if \cref{rr:outside} has been exhaustively applied, then $\deg_{H|0}(v)=0$ implies that $v$~has exactly one neighbor, which is in the 2-core of the graph.

\begin{rrule}\label{rr:path}
If $v_{0},v_{1},\dots, v_{\ell},v_{\ell+1}$ is a path in the input graph such that $\ell\geq 3$, $\deg_{H|0}(v_{i})=2$ for every $i\in \{1, \ldots, \ell\}$, $\deg_{H|0}(v_0) \ge 2$, and $\deg_{H|0}(v_{\ell+1}) \ge 2$, then %
let $r = \min\{ \deg_G(v_i) \mid i \in \{1, \ldots, \ell\}\} - 2$ and remove vertices $v_1, \ldots, v_{\ell}$ and their neighbors not in the 2-core. 
Then introduce two new vertices~$u_{1}$ and~$u_{2}$ with edges $\{v_{0},u_{1}\}$, $\{u_{1},u_{2}\}$, and $\{u_{2},v_{l+1}\}$ and $2r$ further new vertices and connect $u_1$ with $r$ of them and $u_2$ with another $r$ of them. 
\end{rrule}

\begin{proof}[Proof of correctness]
 Let $F$ be a minimal secluded feedback vertex set in~$G$, and let $G'$ be the graph obtained from $G$ by applying~\Cref{rr:path}.
  Suppose $F\cap \{v_1,\ldots,v_\ell\}\neq \emptyset$.
  Since $\deg_{H|0}(v_i)=2$ for all $i\in[\ell]$, each of the vertices $v_1,\ldots,v_\ell$ participates in the same set of cycles of~$G$.
  Hence, it follows that $F\cap \{v_1,\ldots,v_\ell\}=\{v_q\}$ for some $q\in[\ell]$.
  Moreover, the set of cycles where $v_1,\ldots,v_\ell$ appear in is a subset of the set of cycles where $v_0$ appears in and a subset of the set of cycles where $v_{\ell+1}$ appears in.
  Hence, due to minimality of $F$ we have that $v_q \in F$ implies $v_0\not\in F$ and $v_{\ell+1}\not\in F$.
  Due to the definition of $r$ the number of neighbors of $v_q$ not in the 2-core is at least $r$.
  Then $F'=(F\setminus \{v_q\})\cup \{u_1\}$ is a secluded feedback vertex set of $G'$ with $|F'|=|F|$ and $|N_G[F]|\ge|N_{G'}(F')|$.
  
  Suppose $F\cap \{v_1,\ldots,v_\ell\}= \emptyset$ but $F\cap \{v_0,v_{\ell+1}\}\neq \emptyset$.
  Then $|F\cap \{v_0,v_{\ell+1}\}|=|N_G[F] \cap \{v_1,v_\ell\}|=|N_{G'}[F] \cap \{u_1,u_2\}|$.
  It follows that $F$ is a secluded feedback vertex set in $G'$ with $|N_{G'}[F]|=|N_G[F]|$.
  
  The case where $F\cap \{v_0,\ldots,v_{\ell+1}\}= \emptyset$ is trivial.
  
  Conversely, let $F$ be a minimal secluded feedback vertex set in~$G'$.
  Suppose that $F\cap \{u_1,u_2\}\neq \emptyset$.
  Since $F$ is minimal, either $u_1$ or $u_2$ is contained in $F$, since both vertices participate in the same set of cycles in $G'$.
  Without loss of generality, let $u_1\in F$.
  Moreover, $F\cap \{v_0,v_\ell\}= \emptyset$, as otherwise $F\setminus \{u_1\}$ is a smaller secluded feedback vertex set in~$G'$, contradicting the minimality of~$F$.
  By the choice of $r$, there exists $q\in [\ell]$ such that $\deg_G(v_q)- 2=r$.
  Then $F':=(F\setminus \{u_1\})\cup \{v_q\}$ is a feedback vertex set in $G$ with $|N_G[F']|=|N_{G'}[F]|$.

  Suppose that $F\cap \{v_0,v_{\ell+1}\}\neq \emptyset$.
  Since $F$ is minimal, it follows that $F\cap \{u_1,u_2\}=\emptyset$.
  Observe that $F$ is also a feedback vertex set in $G$, as $v_0$ and $v_{\ell+1}$ participate in each cycle containing any vertex in $\{v_1,\ldots,v_\ell\}$.
  Since $|F\cap \{v_0,v_{\ell+1}\}|=|N_{G'}[F]\cap\{u_1,u_2\}|=|N_{G}[F]\cap\{v_1,v_\ell\}|$, it follows that $|N_G[F]|=|N_{G'}[F]|$.
  Hence, $F$ is a secluded feedback vertex set in $G$.
  
  The case where $F\cap \{v_0,u_1,u_2,v_{\ell+1}\}= \emptyset$ is trivial.
\end{proof}

\noindent For $x\in V(G)$, we denote by $\petal(x)$ the maximum number of cycles only intersecting in~$x$.

\begin{rrule}\label{rr:flower}
If there is a vertex $x\in V(G)$ such that $\petal(x)\geq \lceil\frac{k}{2}\rceil$, then output that $(G,k)$ is a no-instance of SFVS.
\end{rrule}

\begin{proof}[Proof of correctness]
There are at least $\lceil\frac{k}{2}\rceil$ cycles in $G$, which 
are vertex-disjoint except for~$x$.
  Assume that $G$ allows a feedback vertex set $F$ with $|N_G[F]|\leq k$.
  Clearly, $F$ must contain at least one vertex in each of the cycles.
  Therefore $N[F]$ must contain at least three vertices of each cycle.
  As only $x$ can be shared among these triples, we get $|N_G[F]|\geq 2\cdot \lceil\frac{k}{2}\rceil + 1>k$.
  It follows that $G$ does not admit a secluded feedback vertex set.
\end{proof}

\begin{rrule}\label{rr:bigdeg}
If~$v\in V(G)$ is a vertex such that $\deg_G(v) >k$, but $\deg_{H|0}( v) < \deg_G( v)$, then remove one of its neighbors not in the 2-core.
\end{rrule}

\begin{proof}[Proof of correctness]
 First observe that, as $\deg_G(v) >k$, vertex $v$ cannot be contained in any secluded feedback vertex set.
  As additionally $\deg_{H|0}( v) < \deg_G( v)$, we know that there is a vertex $w\in N(v)\setminus V(H)$. 
  Since $w$ is not in the 2-core, it is not involved in the cycles of~$G$.
  Since $\deg_G(v) >k$, removing $w$ from~$G$ results in $\deg_{G-\{w\}}(v)\geq k$ and hence, $v$ cannot be contained in any secluded feedback vertex set of~$G-\{w\}$.
  Altogether, $G$ has a feedback vertex set $F$ with $|N_G[F]|\leq k$ if and only if $G-\{w\}$ has a feedback vertex set $F'$ with $|N_{G-\{w\}}[F']|\leq k$.
\end{proof}

\begin{rrule}\label{rr:parallel}
 Let $x,y$ be two vertices of $G$. If there are at least $k$ internally vertex disjoint paths of length at least 2 between $x$ and $y$ in $G$, then output that $(G,k)$ is a no-instance of SFVS. 
\end{rrule}

\begin{proof}[Proof of correctness]
 Observe then that if neither~$x$ nor~$y$ belong to a feedback vertex set~$D$ of~$G$ we need at least~$k-1$ vertices to hit all the cycles, since otherwise there are at least two distinct paths $P_1,P_2$ of length at least 2 between $x$ and $y$ with $(V(P_{1})\cup V(P_{2}))\cap D=\emptyset$ and thus the graph induced by $V(P_{1})\cup V(P_{2})\cup \{x,y\}$ contains a cycle.
    Since each of the $k-1$ vertices has at least two vertices in its open neighborhood and only the vertices $x$ and~$y$ can be shared among these, the closed neighborhood contains at least $k+1$ vertices.
    On the other hand, the open neighborhood of both $x$ and $y$ contains one vertex from each of the $k$ paths. 
    Therefore, their closed neighborhood is of size at least $k+1$ and they cannot be included in the solution.
\end{proof}

\noindent
Note that~\Cref{rr:outside,rr:path,rr:bigdeg,rr:parallel} can be applied trivially in polynomial time.
\Cref{rr:flower} can be applied exhaustively in polynomial time due to the following.
{
\begin{proposition}[\cite{Thomasse10}]\label{algosepvrt}
Let~$G$ be a graph and~$x$ be a vertex of~$G$. 
In polynomial time we can either find a set of $\ell+1$ cycles only intersecting in $x$ (proving that $\petal(x)\geq \ell+1$) or a set of vertices $Z\subseteq V(G)\setminus \{x\}$ of size at most $2\ell$ intersecting every cycle containing~$x$.
\end{proposition}

\noindent An instance $(G,k)$ of SFVS is called {\em reduced} if none of the \Cref{rr:outside,rr:path,rr:flower,rr:bigdeg,rr:parallel} can be applied.
Following the proof by~\citet{GiannopoulouLSS16}, we first give structural decomposition lemma, then bound the size of components of the decomposition, and finally bound the number of components in the decomposition to obtain the polynomial kernel for SFVS parameterized by~$k$.
We start with the following structural decomposition lemma, which identifies the set~$B$.

\begin{lemma}\label[lemma]{lem:setCm}
There is a polynomial time algorithm that given a reduced instance $(G,k)$ of SFVS 
either correctly decides that $(G,k)$ is a no-instance or
finds two sets~$F$ and~$M'$ such that, denoting $B = F \cup M'$, the following holds.
\begin{enumerate}[(i)]
\item $F$ is a feedback vertex set of~$G$.
\item Each connected component of $G - B$ has at most 2 neighbors in~$M'$.
\item For every connected component $C$ in $G- B$ and $x\in B$,  $|N_{G}(x)\cap C|\leq 1$, that is, 
every vertex~$y$ of~$F$ and every vertex~$x$ of~$M'$ have at most one neighbor in every connected component~$C$ of~$G- B$.
\item $|B| \le 4k^2+2k$.
\end{enumerate}
\end{lemma}

{
  Similar to~\citet{GiannopoulouLSS16}, we also make use of the following concept.
  For a rooted tree~$T$ and vertex set~$M$ in~$V(T)$ the lowest common ancestor-closure ({\em LCA-closure})~$\lcac(M)$ is obtained by the following process.
  Initially, set~$M'=M$. 
  Then, as long as there are vertices~$x$ and~$y$ in~$M'$ whose lowest common ancestor~$w$ is not in~$M'$, add~$w$ to~$M'$. 
  Finally, output $M'$ as the LCA-closure of~$M$.

  \begin{lemma}[\citet{FominLMS12}]\label{caclsbnd}
  Let $T$ be a tree and $M\subseteq V(T)$. If $M'=\lcac(M)$ then $|M'|\leq 2|M|$ and for every connected component~$C$ of $T- M'$, $|N_T(C)|\leq 2$.
  \end{lemma}

  We continue with proving our structural decomposition lemma.

  \begin{proof}[Proof of \Cref{lem:setCm}]
  Note that if there is a feedback vertex set of $G$ with closed neighborhood of size at most $k$, then it is also a feedback vertex set in~$G$ of size at most $k$.
  Thus, we can apply the 2-approximation algorithm for \textsc{Feedback Vertex Set} on $G$ due to~\citet{BafnaBF99} to find 
  in polynomial time a feedback vertex set~$F$ of~$G$. 
  If $|F|> 2k$, then we output that $(G,k)$ is a no-instance of SFVS. 
  Hence, we assume $|F|\leq 2k$ in the following.
  Since $F$ is a feedback vertex set in~$G$, property (i) is trivially fulfilled.
  Moreover, $G-F$ is a collection of trees $T_1,\ldots,T_\ell$. 
  We select for each of the trees~$T_i$ some root vertex $v_i\in V(T_i)$.
  It remains to construct the set $M'$ such that $F\cup M'$ fulfills conditions (ii)--(iv).

  Recall that the instance $(G,k)$ is reduced.
  Hence, \Cref{rr:flower} is not applicable, and hence $\petal(x)<\lceil\frac{k}{2}\rceil$ for all $x\in F$.
  We apply \Cref{algosepvrt} to each vertex in $v\in F$, obtaining a set $Z_v\subseteq V(G)\setminus \{v\}$ intersecting each cycle containing $v$ with $|Z_v|\leq k$.
  Let $Z:=Z_1\cup \ldots \cup Z_{|F|}$ denote the union of these sets. 
  Observe that $|Z|\leq 2k^2$.
  We set $M_i := T_i\cap Z$ and $M_i':=\lcac(M_i)$ for all $i\in[\ell]$.
  Observe that, by~\Cref{caclsbnd}, $|M_i'|\leq 2|M_i|$. 
  Finally, we set $M'=\bigcup_{i\in[\ell]} M_i'$ and $B=F\cup M'$ (note that $F\cap M'=\emptyset$).
  Observe that $|M'|\leq \sum_{i\in[\ell]} |M_i'| \le \sum_{i\in[\ell]} 2|M_i|\leq 2|Z| \leq 4k^2$ and by~\Cref{caclsbnd}, for every connected component $C$ in $G-B$ it holds that $|N_{G-F}(C)|\leq 2$ (hence, property (ii) is fulfilled).
  Altogether, $|B|=|F|+|M'|\leq 2k+4k^2$, yielding property (iv).
  It remains to show that property (iii) is fulfilled.

  Let $C$ be a connected component of $G-B$ and $x\in B$ some vertex. 
  Suppose that $x$ has two neighbors in $C$.
  Then $C_x:=C \cup \{x\}$ induces a cycle in $G$ as $C$ is connected.
  If $x\in F$, then this contradicts the set $Z_x \subseteq Z \subseteq M'\cup (F \setminus\{x\})$ hitting every cycle containing $x$.
  If $x\in M'$, then this contradicts the set $F$ hitting each cycle in $G$.
  Hence, property~(iii) is fulfilled.
  \end{proof}
}

\noindent Next, we show that if~$B$ is as in~\Cref{lem:setCm}, then the size and the number of the connected components in the $G-B$ is polynomially bounded from above in the size $k$ of the closed neighborhood of the feedback vertex set in question.
We first bound from above the size of the connected components in $G-B$ as follows.

\begin{lemma}\label[lemma]{lem:compsize}
 Let $(G,k)$ and~$B$ be as in~\Cref{lem:setCm}, and let~$C$ be a connected component of~$G - B$. Then the number of vertices~$|V(C)|$ of the connected component~$C$ is at most $(12k+7)(k+1)$.
\end{lemma}

{
  \begin{proof}
  Le $H$ be the 2-core of $G$. We distinguish two cases on the size of $C_H:=V(C)\cap V(H)$, namely $|C_H|=0$ on the one hand, and $|C_H|>0$ on the other hand.
  
  \emph{Case $|C_H|=0$}:
  Observe that $C$ is a connected component in~$G-V(H)$.
  Hence, by~\Cref{obs:2core}, there is at most one vertex in $C$ adjacent to $H$.
  If $x\in V(C)$ is adjacent to $H$, no other vertex of $C$ is adjacent to~$H$. 
  Suppose that $|V(C)|>1$.
  Since $C$ is connected, there is a vertex $u\in V(C)$ such that $N[u]\subseteq G - V(H)$. 
  Existence of such vertex would contradict the instance being reduced with respect to~\Cref{rr:outside}. 
  Hence, $|V(C)|\leq 1$.
  
  \emph{Case $|C_H|>0$}:
    Recall that $(G,k)$ is reduced.
    On the one hand, due to~\Cref{rr:outside}, we know that every vertex in $C- V(H)$ has a neighbor in $C_H$.
    On the other hand, due to~\Cref{rr:bigdeg}, each vertex in $C_H$ has at most $k$ neighbors in $C- V(H)$.
    Hence, it follows that $|V(C)|\leq (k+1)\cdot |C_H|$.
    Consequently, it remains to bound the number of vertices in~$C_H$.
    
    In the following we count the number of vertices in~$G[C_H]$ having degree 1, 2, and at least 3 in~$G[C_H]$.
    Let $D_H^1\subseteq C_H$ be the set of vertices in~$G[C_H]$ having degree exactly one.
    Since $D_H^1\subseteq V(H)$, it holds that $\deg_{H|0} (v)\geq 2$ for each $v\in D_H^1$. 
    Since there is exactly one neighbor of $v$ in $G[C_H]$, at least one other neighbor is contained in $V(H) \cap B$.
    Let $B_C$ denote the vertices of $C$ having at least one neighbor in~$B$.
    Note that $D_H^1\subseteq B_C$.
    Due to~\Cref{lem:setCm}(ii), $C$ has at most two neighbors in $M'$ (recall $B=F\cup M'$).
    Moreover, due to~\Cref{lem:setCm}(iii), each vertex in $B$ has at most one neighbor in~$C$.
    It follows that $|B_C|\leq|F|+2\leq 2k+2$, and hence $|D_H^1|\leq 2k+2$. 
    
    Let $D_H^{\geq 3}\subseteq C_H$ be the set of vertices in~$G[C_H]$ having degree at least three. 
    Since $G[C_H]$ is acyclic (recall that $F\subseteq B$ is a feedback vertex set), it follows that $D_H^1$ forms the leaves in $G[C_H]$.
    A basic observation on trees is that the number of inner vertices of degree at least three is at most the number of leaves minus one.
    Hence, $|D_H^{\geq 3}|\leq|D_H^1|-1\leq 2k+1$.
    
    Let $D_H^{-2}:=B_C\cup D_H^{\geq 3}$. 
    Observe that $C_H\setminus D_H^{-2}$ only contains vertices having degree exactly two in~$G[C_H]$.
    Moreover, these vertices participate only in paths connecting vertices in~$D_H^{-2}$.
    Since $|D_H^{-2}|\leq 2k+2+2k+1=4k+3$, and $G[C_H]$ is acyclic, there are at most $4k+3-1=4k+2$ many of these paths.
    Moreover, due to~\Cref{rr:path}, these paths contain at most two vertices not being the endpoints. 
    Hence, $|C_H|\leq |C_H\setminus D_H^{-2}| + |D_H^{-2}|\leq 2\cdot (4k+2)+4k+3=12k+7$.
    It follows that $|V(C)|\leq (k+1)\cdot |C_H|\leq (k+1)\cdot(12k+7)$.    
  \end{proof}
}

Having an upper bound on the sizes of the set $B$ and of each connected component in $G-B$, it remains to count the number of connected components in $G-B$.
{
\begin{remark}
It is easy to polynomially upper-bound the number of connected components in $G-B$.
To this end, first observe that by~\Cref{rr:outside}, each connected component in $G-B$ has at least one neighbor in~$B$.
Next, consider those connected components in $G-B$ having exactly one neighbor in $B$.
Due to~\Cref{rr:flower}, each vertex in $B$ is incident to at most $k$ connected components in $G-B$ having exactly one neighbor in $B$.
Hence, the number of these connected components in $G-B$ is upper bounded by $|B|\cdot (k+1)$.
Last, consider those connected components in $G-B$ having at least two neighbors in $B$.
Then it follows from~\Cref{rr:parallel} that two vertices in $B$ are together contained in the neighborhood of at most $k$ connected components in $G-B$. Indeed, each connected component~$C$ of $G - B$ with $\{x,y\}\subseteq N_{G}(C)$ provides a separate path between~$x$ and~$y$. 
Altogether, the number of connected components in $G-B$ is upper bounded by $|B|^2\cdot(k+1)\in O(k^5)$.
Hence, together with~\Cref{lem:compsize}, we obtain a polynomial kernel of size $O(k^7)$ for SFVS.
\end{remark}
}
With the next lemma, we give an $O(k^3)$ upper bound on the number of connected components in~$G-B$.

\begin{lemma}\label[lemma]{lem:setcg}
 Let $(G,k)$ and~$B$ be as in~\Cref{lem:setCm}. Then the number of connected components in~$G - B$ is at most $15k^3+8k^2-k-1$.    
\end{lemma}

{
  \begin{proof}
  We partition the connected components of $G-B$ by the number of their neighbors in~$B$, namely having exactly one neighbor and having at least two neighbors in~$B$.
  For $x\in B$, denote by $B_x$, the set of connected components in $G-B$ having vertex~$x$ as their only neighbor in~$B$.
  Further, for $x,y\in B$, denote by~$B_{xy}$, the set of connected components having at least $x$ and $y$ as their neighbors in~$B$.
  Observe that the connected components of $G-B$ are exactly $\bigcup_{\{x,y\}\subseteq B} (B_x\cup B_{xy})$, and hence the number of the connected components of $G-B$ is at most $|\bigcup_{x\in B} B_x|+|\bigcup_{\{x,y\}\subseteq B} B_{xy}|$.
  Further observe that $|\bigcup_{x\in B} B_x|\leq |B|k\leq 4k^3+2k^2$. 
  Hence, it remains to upper-bound the cardinality of~$\bigcup_{\{x,y\}\subseteq B} (B_{xy})$.
  To this end, observe that 
  \begin{align}
    \bigcup_{\{x,y\}\subseteq B} (B_{xy}) = \underbrace{\bigcup_{\{x,y\}\subseteq F} (B_{xy})}_{:=B^1} \cup \underbrace{\bigcup_{x\in F,y\in M'} (B_{xy})}_{:=\tilde{B}^2} \cup \underbrace{\bigcup_{\{x,y\}\subseteq M'} (B_{xy})}_{:=B^3}.
  \end{align}
  Notice that the equality is still true if we replace $\tilde{B}^2$ by $B^2:=\tilde{B}^2\setminus B^3$, since $B^3$ appears in the union on the right hand-side.
  Hence, in the remainder of this proof, we upper-bound the size of the sets $B^1$, $B^2$, and~$B^3$.
  Observe that the size of $B^1$ is upper bounded by $\binom{2k}{2}(k+1)=2k^3+k^2-k$.
  
  Next, we upper-bound the size of $B^2$.
  To this end, let $x\in F$ be arbitrary but fixed.
  Consider the set $S_x$ of vertices in $M'$ such that there are at least two connected components of $G-B$ neighboring with both $x$ and $y$.
  Observe that for each $y\in S_x$, the set of connected components in $B^2$ neighboring with both $x$ and $y$ is unique, as otherwise there is a connected component in~$B^2$ containing two vertices in $M'$ and hence belonging to $B^3$, contradicting our definition of $B^2:=\tilde{B}^2\setminus B^3$.
  Since for each $y\in S_x$ there are at least two connected components in $B^2$, they together with $x$ and $y$ form a cycle in $G$.
  Hence, due to~\Cref{rr:flower}, the number of vertices in $S_x$ is at most $k/2$.
  On the other hand, there are at most $k$ connected components neighboring with both $x$ and $y$ for any $y\in S_x$ due to~\Cref{rr:parallel}, since each such component provides a separate path of length at least 2 between $x$ and $y$.
  Finally, observe that the number of vertices $y\in M'$ such that there is at most one connected component of $G-B$ neighboring with both $x$ and $y$ is trivially bounded by $|M'|\leq 4k^2$.
  Altogether, we obtain that $|B^2|\leq \sum_{x\in F} (4k^2 + (k/2)(k+1))\leq 2k(4k^2 + (k/2)(k+1))=9k^3+k^2$.

  Last, we upper-bound the size of $B^3=\bigcup_{\{x,y\}\subseteq M'}$.
  Observe that due to~\Cref{lem:setCm}(ii), for each $x,y \in M'$ each connected component~$C$ in $B_{xy}$ only neighbors with $x$ and $y$ out of $M'$, that is, $N(C)\cap M'=\{x,y\}$.
  Moreover, by the connectedness of $C$, $x$ and $y$ are connected via a path through~$C$.
  By known facts on forests and trees, we know that if there are at least $r$ paths connecting vertex pairs out of $r$ vertices in a graph, then there is a cycle in the graph.
  Hence, since $F$ is a feedback vertex set in~$G$, there are at most $|M'|-1$ connected components in $B^3$.
  Recalling that $|M'|\leq 4k^2$, we obtain that $|B^3|\leq |M'|-1 \leq 4k^2-1$.
  
  Altogether, the number of connected components in $G-B$ is at most
  \begin{align*}
    4k^3+2k^2 + |\bigcup_{\{x,y\}\subseteq B} B_{xy}| &\leq 4k^3+2k^2 + |B^1| + |B^2|+ |B^3| \\
    &\leq 4k^3+2k^2 + 2k^3+k^2-k + 9k^3+k^2 + 4k^2-1 \\
    &= 15k^3+8k^2-k-1 .\qedhere
  \end{align*}
  \end{proof}
}

Finally, putting all together, we can prove the the main result of this section.

{
\begin{proof}[Proof of~\Cref{thm:sfvspolyker}]
 Let $(G',k)$ be the input instance of SFVS.
 We compute the 2-core~$H$ of $G$.
 We apply~\Cref{rr:outside,rr:path,rr:flower,rr:bigdeg,rr:parallel} exhaustively to obtain an equivalent instance $(G,k)$ such that $(G,k)$ is reduced.
 Next we apply~\Cref{lem:setCm} and obtain the set $B$ in $G$ with $|B|\leq 4k^2+2k$.
 Let $\cc$ denote the set of connected components in $G-B$.
 By~\Cref{lem:setcg}, we know that $|\cc|\leq 15k^3+8k^2-k-1$.
 Moreover, due to~\Cref{lem:compsize}, for each $C\in\cc$ it holds that $|V(C)|\leq (k+1)\cdot(12k+7)$.
 If follows that the number of vertices $|V(G)|$ in $G$ is at most
 $|B|+|\cc|\cdot \max_{C\in\cc}|V(C)| \leq 4k^2+2k + (15k^3+8k^2-k-1)\cdot(k+1)\cdot(12k+7) \in O(k^5)$.
\end{proof}
}
}

\subsection{Small Secluded Feedback Vertex Set}\label{ssec:ssfvs}

In contrast to restricting the closed neighborhood of a feedback vertex set, restricting the open neighborhood by a parameter yields a W[1]-hard problem.

\decprob{Small Secluded Feedback Vertex Set}
	{A graph~$G=(V,E)$ and two integers $k, \ell$.}
	{Is there a set~$S\subseteq V$ such that $G-S$ is cycle-free, $|S|\leq k$, and $|N_G(S)|\leq \ell$?}

\begin{theorem}
  \label{thm:ssfvsW1hard}
 \textsc{Small Secluded Feedback Vertex Set} is W[1]-hard with respect~to~$\ell$.
\end{theorem}

{
  \begin{proof}
  We provide a parameterized reduction from \textsc{Multicolored Independent Set (MIS)}: given a $k$-partite graph $G=(V,E)$ and its partite sets~$V_1 \cup \ldots\cup V_k = V$, the question is whether there is an independent set~$I$ of size~$k$ %
  such that $I\cap V_i\neq \emptyset$ for each~$i\in\{1, \ldots,k\}$. MIS is W[1]-hard when parameterized by the size~$k$ of the independent set~\citep{FellowsHRV09}.
  
  Let $G=(V,E)$ with partite sets~$V_1\cup V_2 \cup \ldots \cup V_k = V$ be an instance of MIS.
  We can assume that for each $i\in \{1, \ldots,k\}$ we have $|V_i|\geq 2$ and there is no edge $\{v,w\}\in E$ with $v,w\in V_i$.
  We create an instance~$(G',k',\ell)$ of \textsc{Small Secluded Feedback Vertex Set (SSFVS)} with $k'=|V|-k$ and $\ell=k+1$ as follows.

  \emph{Construction}: 
  (Refer to \cref{fig:ssfvsW1hard} for a sketch of the construction.)
  \begin{figure}
    \centering
    \begin{tikzpicture}
      \usetikzlibrary{shapes}
      \def\ya{1.25}
      \node (u) at (4,3)[circle,draw,scale=2/3,label=0:{$u$}]{};

      \node (u1) at (3,3.75)[circle,fill,scale=1/2,draw]{};
      \node (u2) at (3.5,3.75)[circle,fill,scale=1/2,draw]{};
      \node (uldts) at (4.25,3.75)[]{$\ldots$};
      \node (ukl) at (5,3.75)[circle,fill,scale=1/2,draw]{};

\draw[thick,decorate,decoration={brace,amplitude=5pt}] (3-0.2,3.75+0.2) -- (5+0.2,3.75+0.2) node[midway, above,yshift=5pt,]{$k'+\ell$ vertices in $L$};

      \draw (u) -- (u1);
      \draw (u) -- (u2);
      \draw (u) -- (ukl);

            \foreach \x in {1,2,...,6}{
	      \draw[-] (u) to [out=210,in=70+20-5*\x] (0-1-0.8+0.25*\x,\ya);
	      \draw[-] (u) to [out=225,in=90+20-5*\x] (3-0.70+0.20*\x,\ya);
	      \draw[-] (u) to [out=-30,in=110+20-5*\x] (8+1-0.85+0.25*\x,\ya);
      }

      \node[ellipse,minimum width=50pt,minimum height=20pt,draw, fill=white] (V1)  at (0-1,\ya)[label=135:{$V_1$}]{};
      \node[ellipse,minimum width=50pt,minimum height=20pt,draw, fill=white] (V2)  at (3,\ya)[label=155:{$V_2$}]{};
      \node at (5.5,\ya)[scale=1.75]{$\cdots$};
      \node[ellipse,minimum width=50pt,minimum height=20pt,draw, fill=white] (Vk)  at (8+1,\ya)[label=45:{$V_k$}]{};

      \node (v1i) at (0-0.5,\ya)[fill,circle,scale=1/3,draw,label=-180:{$v^1_i$}]{};
      \node (v2j) at (3,\ya)[fill,circle,scale=1/3,draw,label=-0:{$v^2_j$}]{};

      \draw[-,out=-45,in=-135] (v1i) to node[midway,below]{$\{v^1_i,v^2_j\}\in E(G)$}(v2j);

    \end{tikzpicture}
    \caption{Sketch of the construction of graph $G'$ on an input graph $G=(V = V_1 \cup \ldots\cup V_k, E)$ as used in the proof of \cref{thm:ssfvsW1hard}. The ellipses correspond to cliques with vertex sets $V_i$, $i\in \{1, \ldots,k\}$.}\label{fig:ssfvsW1hard}
  \end{figure}
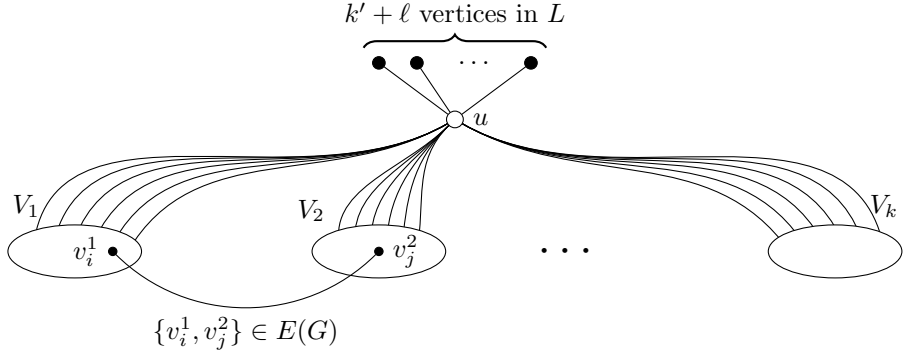
  Initially, let $G':=G$. 
  For each $i\in \{1, \ldots,k\}$ turn $V_i$ into a clique, that is, add the edge sets $\{\{a,b\}\mid a,b\in V_i, a\neq b\}$.
  Next, add to $G'$ a vertex $u$ and a set $L$ of $k'+\ell$ vertices. 
  Finally, connect each vertex in $V \cup L$ to~$u$ by an edge.
  
  \emph{Correctness}: 
  We show that $(G,k)$~is a \yes-instance of MIS if and only if $(G',k',\ell)$~is a \yes-instance of SSFVS.
  
  \raproof{} 
  Let $(G,k)$ be a \yes-instance of MIS and let $I\subseteq V$ with $|I|=k$ be a multicolored independent set in~$G$.
  We delete all vertices in~$S:=V(G')\setminus (I \cup L \cup \{u\})$ from~$G'$.
  Observe that $|S|=|V|-k=k'$.
  Moreover, $N_{G'}(S) = k+1=\ell$.
  Since there is no edge between any two vertices in~$I$, $G-S$~forms a star with center~$u$ and $k'+\ell+1+k$ vertices.
  Since every star is acyclic, $(G',k',\ell)$~is a \yes-instance of SSFVS.
  
  \laproof{} 
  Let $(G',k',\ell)$ be a \yes-instance of SSFVS and let $S\subseteq V(G')$ be a solution.
  Observe that $G'[V_i\cup\{u\}]$ forms a clique of size $|V_i|+1$ for each $i\in \{1, \ldots,k\}$.
  Since the budget does not allow for deleting the vertex $u$ (i.e.~$u\not\in S$), all but at most one vertex in each $V_i$ must be deleted.
  Since $k'=|V|-k$ and $|V_i|\geq 2$ for all~$i\in \{1, \ldots,k\}$, $S$ contains exactly $|V_i|-1$ vertices of $V_i$ for each $i\in \{1, \ldots,k\}$.
  Hence, $|S|=|V|-k$ and $N_{G'}(S)= k+1=\ell$.
  Let $F:=V\setminus S$ denote the set of vertices in $V$ not contained in $S$.
  Recall that $|F|=k$ and $|F\cap V_i|=1$ for all~$i\in\{1, \ldots,k\}$.
  Next, suppose there is an edge between two vertices $v,w\in F$.
  Since $u\not\in S$ and $u$~is incident to all vertices in~$V$, the vertices~$u,v,w$ form a triangle in $G'$.
  This contradicts the fact that $S$~is a solution for~$(G',k',\ell)$, that is, that $G'-S$ is acyclic.
  It follows that $E(G'[F])=\emptyset$, that is, no two vertices in $F$ are connected by an edge.
  Together with $|F|=k$ and $|F\cap V_i|=1$ for all $i\in\{1, \ldots,k\}$, it follows that $F$ forms a multicolored independent set in $G$.
  Thus, $(G,k)$ is a \yes-instance of MIS.
  \end{proof}
}

\section{Independent Set}\label{sec:is}

For \textsc{Independent Set}, it makes little sense to bound the size of the closed neighborhood from above, as in this case the empty set always constitutes a solution. One might ask for an independent set with closed neighborhood as large as possible. However, for any inclusion-wise maximal independent set~\(S\), one has \(N[S]=V\). Hence, this question is also trivial.
Therefore, in this section we only consider the following problem.

\decprob{Large Secluded Independent Set (LSIS)}
	{A graph~$G=(V,E)$ and two integers $k,\ell$.}
	{Is there an independent set $S\subseteq V$ such that $|S|\geq k$ and $|N_G(S)|\leq \ell$?}

\noindent
The case $\ell = |V|$ equals \textsc{Independent Set} and, thus, \textsc{LSIS} is $W[1]$-hard with respect to $k$. We show that \textsc{LSIS} is also W[1]-hard when parameterized by $k+\ell$.	

\begin{theorem}
  \label{thm:lsishard}
\textsc{Large Secluded Independent Set} is $W[1]$-hard with respect to~$k+\ell$.
\end{theorem}

We remark that the proof of \cref{thm:lsishard} is identical to the W[1]-hardness proof for \textsc{Cutting $\ell$ Vertices}~\citep{Marx06}.
However, for the sake of completeness, we present the proof in the remainder of this section.
{
\begin{proof}
  We provide a polynomial-parameter transformation from \textsc{Clique} parameterized by the solution size~$k$.
  
  \emph{Construction.} 
  Let $(G,k)$~be an instance of~\textsc{Clique} and assume without loss of generality that $k<|V(G)|-1$ (otherwise, solve the instance in polynomial time). We construct an equivalent instance $(G',k',\ell')$ of \textsc{Large Secluded Independent Set} as follows (see~\cref{fig:lsisconstr} for an example).
  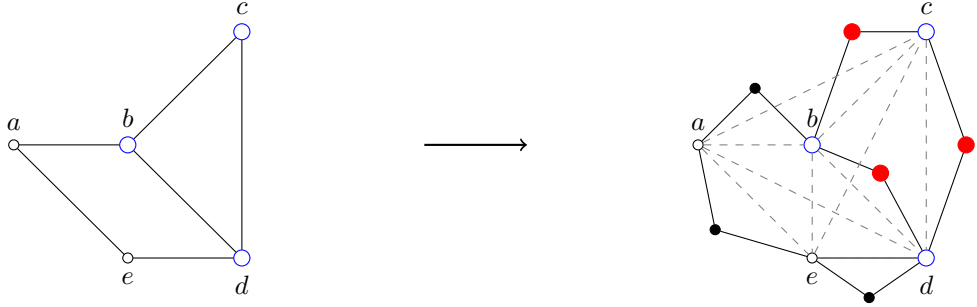
\begin{figure}
    \centering
    \begin{tikzpicture}[scale=1.5]

      \tikzstyle{xnode}=[circle, scale=0.4,draw]

      \node (a) at (0,0)[xnode, label=90:{$a$}]{};
      \node (b) at (1,0)[xnode, label=90:{$b$}, color=blue, scale=1.6]{};
      \node (c) at (2,1)[xnode, label=90:{$c$}, color=blue, scale=1.6]{};
      \node (d) at (2,-1)[xnode, label=-90:{$d$}, color=blue, scale=1.6]{};
      \node (e) at (1,-1)[xnode, label=-90:{$e$}]{};

      \draw (a) --(b);
      \draw (b) --(c);
      \draw (b) --(d);
      \draw (c) --(d);
      \draw (a) --(e);
      \draw (d) --(e);

      \def\x{6};
      \draw[thick,->] (0.6*\x,0) to (0.75*\x,0);
      \tikzstyle{xxnode}=[circle, scale=0.4, fill, draw]

      \node (a) at (0+\x,0)[xnode, label=90:{$a$}]{};
      \node (b) at (1+\x,0)[xnode, label=90:{$b$}, color=blue, scale=1.6]{};
      \node (c) at (2+\x,1)[xnode, label=90:{$c$}, color=blue, scale=1.6]{};
      \node (d) at (2+\x,-1)[xnode, label=-90:{$d$}, color=blue, scale=1.6]{};
      \node (e) at (1+\x,-1)[xnode, label=-90:{$e$}]{};

      \node (ab) at (\x+0.5, 0.5)[xxnode]{};
      \draw (a) -- (ab) -- (b);
      \node (bc) at (\x+1.35, 1.0)[xxnode, color=red, scale=1.6]{};
      \draw (b) -- (bc) -- (c);
      \node (bd) at (\x+1.6, -0.25)[xxnode, color=red, scale=1.6]{};
      \draw (b) -- (bd) -- (d);
      \node (cd) at (\x+2.35, 0)[xxnode, color=red, scale=1.6]{};
      \draw (c) -- (cd) -- (d);
      \node (ae) at (\x+0.15, -0.75)[xxnode]{};
      \draw (a) -- (ae) -- (e);
      \draw (d) --(e);
      \node (de) at (\x+1.5, -1.35)[xxnode]{};
      \draw (d) -- (de) -- (e);

      \draw[color=gray, thin, dashed] (a) --(b);
      \draw[color=gray, thin, dashed] (b) --(c);
      \draw[color=gray, thin, dashed] (b) --(d);
      \draw[color=gray, thin, dashed] (c) --(d);
      \draw[color=gray, thin, dashed] (a) --(e);
      \draw[color=gray, thin, dashed] (d) --(e);
      \draw[color=gray, thin, dashed] (a) --(c);
      \draw[color=gray, thin, dashed] (e) --(c);
      \draw[color=gray, thin, dashed] (b) --(e);
      \draw[color=gray, thin, dashed] (a) --(d);
    \end{tikzpicture}
    \caption{Example of the construction used in the proof of \Cref{thm:lsishard}.  The left-hand side shows the original graph, the right-hand side the graph constructed by the reduction, where the newly introduced edges between each pair of vertices of the original graph are drawn in dashed grey. The vertices introduced for each edge of the original graph are filled red and black, the corresponding new edges are drawn in black. Note that the enlarged, blue vertices of the original graph form a clique and that the vertices corresponding to the edges of said clique (enlarged, filled  red) form an independent set in the new graph.}\label{fig:lsisconstr}
  \end{figure}
  Initially, let $G'$~be an empty graph. 
  Add all vertices of $G$ to~$G'$. 
  Denote the vertex set by~$V$.
  If two vertices of~$G$ are adjacent, we add a vertex to $G'$, that is,
  \(G'\)~additionally to~\(V\) contains the vertex set $X:=\{x_{uv}\mid \{u,v\}\in E\}$.
  Next, connect~$x_{uv}$ to~$u$ and~$v$, that is, add the edge set $E'=\{\{u,x_{uv}\},\{v,x_{uv}\}\mid \{u,v\}\in E\}$.
  Finally, connect any two vertices in~$V$ by an edge.
  Graph $G'$~consists of the vertex set $V\cup X$ and of the edge set~$E'\cup \binom{V}{2}$.
  Observe that $X$~forms an independent set in~$G'$.
  Set $k':=\binom{k}{2}$ and $\ell':=k$.
  We claim that $(G,k)$ is a yes-instance of \textsc{Clique} if and only if $(G',k',\ell')$ is a yes-instance of \textsc{Large Secluded Independent Set}.

  \raproof{}
  Let $C\subseteq V(G)$ be a clique of size~$k=|C|$ in $G$.
   We claim that $X':=\{x_{u,v}\mid u,v\in C\}$ forms an independent set of size $\binom{k}{2}$ with $|N(X')|=k=\ell'$ in~$G'$.
  Since $X'\subseteq X$, $X'$ forms an independent set.
  Moreover, since $C$~is a clique of size $k$, there are $\binom{k}{2}$ edges in $G[C]$, and thus $|X'|=\binom{k}{2}$.
  By construction, each vertex in $X$~is only adjacent to vertices in~$C$. Hence, $|N(X')|=|C|=k$.
  Therefore, $X'$~witnesses that $(G',k',\ell')$~is a yes-instance of \textsc{Large Secluded Independent Set}.

  \laproof{}
  Let $U\subseteq V(G')$ form an independent set of size~$k'$ with open neighborhood of size upper-bounded by $\ell'$.
  Suppose that $v\in V$ is contained in $U$ (observe that $U$ contains at most one vertex of $V$, as otherwise it would not be independent).
  Then $|N(U)|\geq |V|-1>k=\ell'$, which contradicts the choice of~$U$.
  It follows that $U\cap V=\emptyset$, and hence $U\subseteq X$.
  By construction, for each $x_{uv}\in U$, the vertices $u,v$ are contained in $N(U)$.
  Since each vertex in $U$ corresponds to an edge in $G$, we have $\binom{k}{2}$ edges incident with at most $k$ vertices.
  The only graph that fulfills this property is the complete graph on $k$ vertices.
  Hence, $G$ contains a clique of size~$k$, and thus $(G,k)$ is a yes-instance of~\textsc{Clique}($k$).
\end{proof}
}

\section{Summary and Future Work}\label{sec:summary}

In this paper, we studied the problem of finding solutions with small neighborhood to
classical combinatorial optimization problems in graphs. We presented computational complexity results for
secluded and small secluded variants of \textsc{$s$-$t$-Separator},
\textsc{$q$-Dominating Set}, \textsc{Feedback Vertex Set},
\textsc{$\mathcal{F}$-free Vertex Deletion}, and for the large
secluded variant of~\textsc{Independent Set}.  
In the case of \textsc{$q$-Dominating Set}, we leave as an open question the parameterized complexity of \textsc{Small $p$-Secluded $q$-Dominating Set}, with $2p>q$, when parameterized by~$\ell$.
Concerning \textsc{Secluded $\mathcal{F}$-free Vertex Deletion}, we
would like to point out that it is an interesting question which
families~$\mathcal{F}$ exactly yield \NP-hardness as opposed to
polynomial-time solvability.

A natural way to generalize our results would be to consider vertex-weighted graphs and directed graphs. This generalization was already investigated by~\citet{Chechikjpp16} for \textsc{Secluded Path} and \textsc{Secluded Steiner Tree}. Furthermore, replacing the bound on the open neighborhood in the case of small secludedness by a bound on the outgoing edges of a solution would be an interesting modification of the problem. The variation follows the idea of the concept of \emph{isolation} \citep{hkmn09,huffnerks15,itoio05,khmn09}. As the number of outgoing edges is at least as large as the open neighborhood, this might offer new possibilities for fixed-parameter algorithms. Finally, we focused on solutions of size at most or at least an integer~$k$ and did not discuss the case of size \emph{exactly}~$k$ so far.%

\paragraph{Acknowledgment}
This research was initiated at the annual research retreat of the
algorithms and complexity group of TU~Berlin, held in Krölpa,
Thuringia, Germany, from April 3\textsuperscript{rd} till April 9\textsuperscript{th}, 2016.

We would like to thank the anonymous referees of IPEC for comments
that helped to improve the paper and for pointing us to the work
of~\citet{FominGK13}.  The second author thanks Nikolay Karpov
(St.~Petersburg Department of the Steklov Institute of Mathematics of
the Russian Academy of Sciences) for discussion on secluded problems.

\section*{\bibname}%
\bibliographystyle{abbrvnat}
\bibliography{secluded-arxiv-finalbib}%

\end{document}